\documentclass[oneside,english,american]{amsart}
\usepackage[T1]{fontenc}
\usepackage[latin9]{luainputenc}
\usepackage{geometry}
\usepackage{amstext}
\usepackage{amsthm}
\usepackage{amssymb}
\usepackage{setspace}
\usepackage{color}
\usepackage{graphicx}
\usepackage{subcaption}
\usepackage{float}
\usepackage[toc,page]{appendix}

\makeatletter
\newtheorem{innercustomthm}{Theorem}
\newenvironment{customthm}[1]
  {\renewcommand\theinnercustomthm{#1}\innercustomthm}
  {\endinnercustomthm}
\numberwithin{equation}{section}
\numberwithin{figure}{section}
 
\theoremstyle{plain}
\newtheorem{thm}{\protect\theoremname}[section]
  \theoremstyle{plain}
  \newtheorem{fact}[thm]{\protect\factname}
  \theoremstyle{definition}
  \newtheorem{defn}[thm]{\protect\definitionname}
  \theoremstyle{remark}
  
  \theoremstyle{plain}
  \newtheorem{lem}[thm]{\protect\lemmaname}
  \theoremstyle{plain}
  \newtheorem{cor}[thm]{\protect\corollaryname}
  \theoremstyle{remark}
  \newtheorem{claim}[thm]{\protect\claimname}
  \theoremstyle{plain}
  
  \theoremstyle{remark}
  \newtheorem{remark}[thm]{\protect\remarkname}
  
 \theoremstyle{plain}
  \newtheorem*{fact*}{\protect\factname}
  \theoremstyle{plain}
  \newtheorem*{thm*}{\protect\theoremname}
  \theoremstyle{plain}
  \newtheorem*{lem*}{\protect\lemmaname}
 \theoremstyle{definition}
 \newtheorem*{defn*}{\protect\definitionname}
  \theoremstyle{remark}
  \newtheorem*{claim*}{\protect\claimname}
  \theoremstyle{plain}
  \newtheorem*{cor*}{\protect\corollaryname}

\makeatother

\usepackage{babel}
  \addto\captionsamerican{\renewcommand{\claimname}{Claim}}
  \addto\captionsamerican{\renewcommand{\corollaryname}{Corollary}}
  \addto\captionsamerican{\renewcommand{\definitionname}{Definition}}
  \addto\captionsamerican{\renewcommand{\factname}{Fact}}
  \addto\captionsamerican{\renewcommand{\lemmaname}{Lemma}}
  \addto\captionsamerican{\renewcommand{\propositionname}{Proposition}}
  \addto\captionsamerican{\renewcommand{\remarkname}{Remark}}
  \addto\captionsamerican{\renewcommand{\theoremname}{Theorem}}
  \addto\captionsenglish{\renewcommand{\claimname}{Claim}}
  \addto\captionsenglish{\renewcommand{\corollaryname}{Corollary}}
  \addto\captionsenglish{\renewcommand{\definitionname}{Definition}}
  \addto\captionsenglish{\renewcommand{\factname}{Fact}}
  \addto\captionsenglish{\renewcommand{\lemmaname}{Lemma}}
  \addto\captionsenglish{\renewcommand{\propositionname}{Proposition}}
  \addto\captionsenglish{\renewcommand{\remarkname}{Remark}}
  \addto\captionsenglish{\renewcommand{\theoremname}{Theorem}}
  \providecommand{\claimname}{Claim}
  \providecommand{\corollaryname}{Corollary}
  \providecommand{\definitionname}{Definition}
  \providecommand{\factname}{Fact}
  \providecommand{\lemmaname}{Lemma}
  \providecommand{\propositionname}{Proposition}
  \providecommand{\remarkname}{Remark}
  \providecommand{\theoremname}{Theorem}
\providecommand{\theoremname}{Theorem}

\newcommand\blfootnote[1]{%
  \begingroup
  \renewcommand\thefootnote{}\footnote{#1}%
  \addtocounter{footnote}{-1}%
  \endgroup
}

\global\long\def\ket#1{\left|#1\right\rangle }

\global\long\def\identity{\mathbf{1}}

\global\long\def\sp{2D-CLH^*(k,2)}

\global\long\def\a{\mbox{}}
\global\long\def\l{\mbox{}}

\global\long\def\k{\mathbb{\mathcal{K}}}
\global\long\def\a{\mathbb{\mathcal{A}}}
\global\long\def\h{\mathbb{\mathcal{H}}}
\global\long\def\s{\mathcal{S}}
\global\long\def\l{\mathcal{L}}

\title{On the Complexity of Two Dimensional Commuting Local Hamiltonians}
\author{Dorit Aharonov$^*$}
\author{Oded Kenneth$^{*\S}$}
\author{Itamar Vigdorovich$^{*\dagger}$}

\begin{document} 
\maketitle
\blfootnote{$*$ Hebrew University of Jerusalem \\ 
$\dagger$ Weizmann Institute of Science \\
$\S$ Technion - Israel Institute of Technology }
\begin{abstract}

The complexity of the commuting local Hamiltonians (CLH) problem 
still remains a mystery after two decades of research of quantum 
Hamiltonian complexity; it is only known to be contained in NP for 
few low parameters. Of particular interest is the tightly related 
question of understanding whether groundstates of CLHs can be generated 
by efficient quantum circuits. 
The two problems touch upon conceptual, physical and 
computational questions, including the centrality of non-commutation in 
quantum mechanics, quantum PCP and the area law. 
It is natural to try to address first the more physical case 
of CLHs embedded on a 2D lattice, but this problem too remained open apart from some very specific cases 
\cite{3local, planarHastings, Schuch}.   
Here we consider a wide class of two 
dimensional CLH instances; these are $k$-local CLHs, for any constant $k$; 
they are defined on qubits set on the edges of
any surface complex, where we require that this surface complex 
is not too far from being ``Euclidean''.
Each vertex and each face can be associated with 
an arbitrary term (as long as the terms commute). 
We show that this class is in NP, and moreover that 
the groundstates have an efficient 
quantum circuit that prepares them. 
This result subsumes that of Schuch \cite{Schuch} which regarded the special case of $4$-local Hamiltonians on a grid with qubits, and by that it removes the mysterious feature of Schuch's proof which showed containment in NP without providing a quantum circuit for the groundstate and considerably generalizes it. We believe this work and the tools we develop make a significant step
towards showing that 2D CLHs are in NP.  
\end{abstract}

\newpage

\tableofcontents

\newpage

\setcounter{page}{1}
\selectlanguage{english}%

\selectlanguage{american}%

\section{Introduction} 
\subsection{Commuting local Hamiltonians}
The Local Hamiltonian (LH) problem is central to the theory
of quantum complexity. In 1998 it was proved by Kitaev to be QMA-complete
\cite{LH}, initiating by that the area of quantum Hamiltonian complexity.
This result is often considered as the quantum analogue
of the celebrated 
Cook-Levin theorem, which states that the Boolean Satisfiability
problem (SAT) is NP-complete \cite{CookLevin-1}. 
In 2003 Bravyi and Vyalyi \cite{BV} raised the question of what is the 
complexity of the intermediate class in which all terms mutually commute 
(commuting local Hamiltonians, or CLHs). 
The question begs an answer not only because the commutation 
restriction is natural and often made in physics; 
but this is also a computational probe to the fundamental question:  
is the uncertainty exhibited by non-commuting operators
necessary for quantum systems to exhibit their {\it full} quantum nature?   
or, perhaps, it happens to be the (much less expected) case
that even commuting quantum systems can express
full quantum power.   

The CLH problem may seem at first sight to be trivially in NP, 
since by the commutation condition, there  
exists a common basis of eigenstates to all terms, where  
each constraint has a well defined value on  
each eigenstate; the problem seems like a 
classical constraint satisfaction problem (CSP). 
This hope breaks down when realizing that 
the eigenstates themselves maybe highly complex. While in CSP,
a proof for satisfiability is simply a string, i.e. a satisfying assignment,
in the quantum case the eigenstates themselves may be highly 
entangled. Indeed, a beautiful example is Kitaev's 
toric code \cite{ToricCode}, whose global entanglement 
is characterized by topological properties.   In the general case, we do not no whether groundstates of CLHs have an efficient classical description at all (that is, a polynomial size classical representation from which  
the result of any local measurement can be deduced efficiently).

The question of CLHs touches upon some of the most important aspects of  
quantum many body systems: fundamental, 
physical and complexity theoretical.  
For a start, stabilizer codes can be viewed as ground spaces of 
CLHs; these constitute by far the most common framework for 
the study of quantum error correcting codes. 
CLHs are also a very convenient 
place to start with when studying open problems and toy examples;  
for example in the study of  
the quantum PCP conjecture 
\cite{Detectability lemma, Naveh, QPCP} 
often CLHs are used as a case study 
(e.g. \cite{Localizable, FreedmanHastings, CLHexpander}). 
Moreover, CLH systems provide the simplest examples for systems obeying the 
{\it area law} bounding the entanglement in groundstates of gapped systems\footnote{the area law states that the entanglement in the 
groundstate between two regions grows like the size of 
the {\it boundary} between these two regions, 
rather than their volume}. In the one dimensional case, the area law 
was recently shown in a breakthrough result to
provide an efficient classical algorithm for constructing groundstates
\cite{LandauVaziraniVidick}.  
In two or higher dimensions such an algorithm cannot be expected, 
since CLHs become NP hard in 2D. 
However it is still possible that groundstates 
satisfying the area law have polynomial size quantum circuits 
(which may be hard to find). 
Understanding whether groundstates of 2D CLH systems have efficient descriptions 
is thus an essential first step towards clarifying how the area law 
affects the complexity of groundstates. 

Despite the importance and fundamental nature of this class,  
and fourteen years after the problem was posed \cite{BV}, 
the complexity of the CLH problem remains a mystery, even in 
the physically motivated case of 2D.  A trivial upper bound to the complexity of the CLH
problem is that it belongs to QMA. A simple lower bound exists as well: if
we let $d$ denote the dimension of the particles, and let $k$ denote the maximal number of particles that each local term acts on, then
we may define $CLH(k,d)$ accordingly. Using this notation 
$CLH(k,d)$ is NP-hard if $k,d \geq 2$.
The question becomes then to distinguish between those cases which are within NP, those which are QMA hard, and possibly, the intermediate cases. 
However, excluding a few special cases of CLH, not much is known.

\subsection{Previous results}
\label{subsec:prev}
Bravyi and Vyalyi proved that $CLH(2,d)$, namely the class of instances in which the 
particle dimensionality 
$d$ is an arbitrary constant, whereas the interactions only involve two 
such particles (this is called {\it two-local CLHs}), is in NP \cite{BV}. 
The proof relies 
on a decomposition lemma based on the theory of finite dimensional
C{*}-algebra representations. This tool has become 
essential in all following results about this problem. 
 
Aharonov and Eldar \cite{3local} then considered the $3$-local case with qubits
and qutrits. They showed that $CLH(3,2)\in NP$ and also that $NE-CLH(3,3)\in NP$
where NE is a geometrical restriction on the interaction called {\it nearly
Euclidean} \cite{3local}. 
An important fact about the proofs for both
of these results is that the witness which is sent by the prover is
virtually a constant depth quantum circuit which prepares a groundstate for
the system, starting from a product state. 
Hastings called states 
which can be generated by constant depth quantum circuits 
``trivial'' \cite{Localizable}; the name is justified since 
indeed, local observables can be computed classically  
in an efficient way for such states, given the circuit that generates them, 
because the {\it light cone} of qubits affecting the output 
qubits of a local observable is of constant size. 
Thus, the above mentioned results 
not only prove containment in NP, but also 
show that such systems have groundstates with very restricted 
multi-particle entanglement which is in some sense {\it local}. 

In this regard, Aharonov and Eldar \cite{3local} mentioned 
a tight ``threshold'' which
can be drawn at this point: commuting systems with parameters as above
are essentially classical; But, when raising $k$ or $d$ just by
$1$, i.e when considering $CLH(4,2)$ or $CLH(3,4)$, 
we arrive at a new regime in which the quantum system
can exhibit {\it global} entanglement, namely, the groundstates 
are no longer trivial (by Hastings' definition). 
In fact, such systems can exhibit global entanglement even when 
the system is embedded on a square lattice:  
Kitaev's toric code \cite{ToricCode} is a wonderful example, as 
it can indeed be shown that groundstates of this code with nearest neighbor interactions cannot be generated 
by a constant depth quantum circuit \cite{TQObounds}. This
raises the possibility \cite{3local} that general CLH systems with parameters 
above the ``transition point'' are too complex for containment in NP, 
as they allow global entanglement.  

There are several examples beyond the transition point which 
indicate that though global entanglement is possible, 
it might still be the case that CLH systems remain "classically accessible" even in that regime. 
First, it is known that despite their global entanglement, 
toric code states can be constructed in {\it logarithmic} 
depth quantum circuits called 
MERA \cite{MERA} which moreover, allow
local measurements to be simulated classically efficiently. In addition,
Schuch proved that CLHs in which all qubits and all $4$-local constraints
are embedded on a square lattice (generalizing the toric 
code to general interactions with the same geometry and dimensionality) 
also belong to NP \cite{Schuch}.
Interestingly enough though, Schuch's proof bypasses the question of whether 
an efficient description of a groundstate exists; 
instead, the witness which is sent by the prover
convinces the verifier that a low energy state exists without describing
that state at all. Schuch's result thus leaves open the possibility, suggested in \cite{3local}, that when crossing the transition point from local to global entanglement mentioned above, groundstates may in general become difficult to describe classically (not including the toric code special case).

Hastings provided two other results  proving upper bounds on the complexity of the CLH problem in certain cases. In \cite{Localizable} he 
considered $k$-local CLHs whose interaction graphs 
are {\it $1$-localizable}; roughly speaking, these are 
instances whose interaction 
graphs can be mapped to graphs continuously, 
such that the preimage of every point is of bounded diameter. 
This extends the result of \cite{BV} that two local Hamiltonians are 
in NP, to slightly more general constructions which are in some sense, 
two-local in every local region.   
In another result of Hastings \cite{planarHastings},  
he considered CLHs on a planar lattice, 
and proved that the problem is in NP under 
certain restrictive conditions on the C*-algebraic 
decomposition (essentially, that when dividing the lattice to stripes, 
the transformation which disentangles adjacent stripes, a'la Bravyi 
and Vyalyi \cite{BV}, is local). Hastings also provided parts of a 
proof that 2D CLH is in NP, and suggested that the proof 
will be completed elsewhere, however this was not done. 

We note that an interesting clue pointing in fact in the other 
direction, namely suggesting that the  
CLH problem could be harder than NP,  
was given recently by Gosset, Mehta and 
Vidick \cite{GSconnectivityQCMAhard}; they show 
that a certain problem regarding the connectivity of the ground 
space of CLHs is as hard as that of general LHs. 
It is suggested in \cite{GSconnectivityQCMAhard} that this 
is probably true even for CLHs in 2D, though this remains to be worked out. 

We are left with the mystery: 
possibly the above "classical" examples are just special cases, and 
in the general case above the low parameters threshold, global entanglement 
prevents an efficient description of the groundstates of CLHs; 
or maybe, the "classicality" of the entanglement in the 
toric code groundstates as well as in the other examples 
mentioned \cite{Schuch, planarHastings} 
is generic for all CLHs, and thus the problem lies in NP. 

\subsection{Results}
\selectlanguage{english}%
 We consider a wide subclass of CLH in 2D. 
Specifically, we consider $CLH(k,2)$ instances 
(i.e with qubits) where the qubits are arranged on the edges of a 
polygonal complex $\k$ whose underlying topological space is a surface. We refer to those as {\it 2D complexes}\footnote{despite some friction with ordinary {\it simplicial 2-complexes} as in e.g \cite{Hatcher} which do not necessarily define topologically a surface}.
The local terms live on the vertices of $\k$ (these are called stars), 
 and on its faces (plaquettes), where each of these terms
acts on the edges attached to the vertex or the face, 
respectively. In Section \ref{sec:formulation},
this class is formally defined and denoted by $2D-CLH^*(k,2)$. We shall
emphasize that the Hamiltonian terms need not be of the form of 
products of $\sigma_x$ or $\sigma_z$ Paulis as in Kitaev's surface codes, but can be general operators on the relevant qubits (as long as they commute).  
Moreover, the locality parameter $k$, which in this case equals the maximal degree of vertices and faces of $\k$ (a degree of a face is the number of its edges), 
is an arbitrary constant as well. 

\begin{minipage}{0.5\textwidth}
\begin{figure}[H]	
\caption{\label{fig:complex} Polygonal complex}
\includegraphics[scale=0.15]{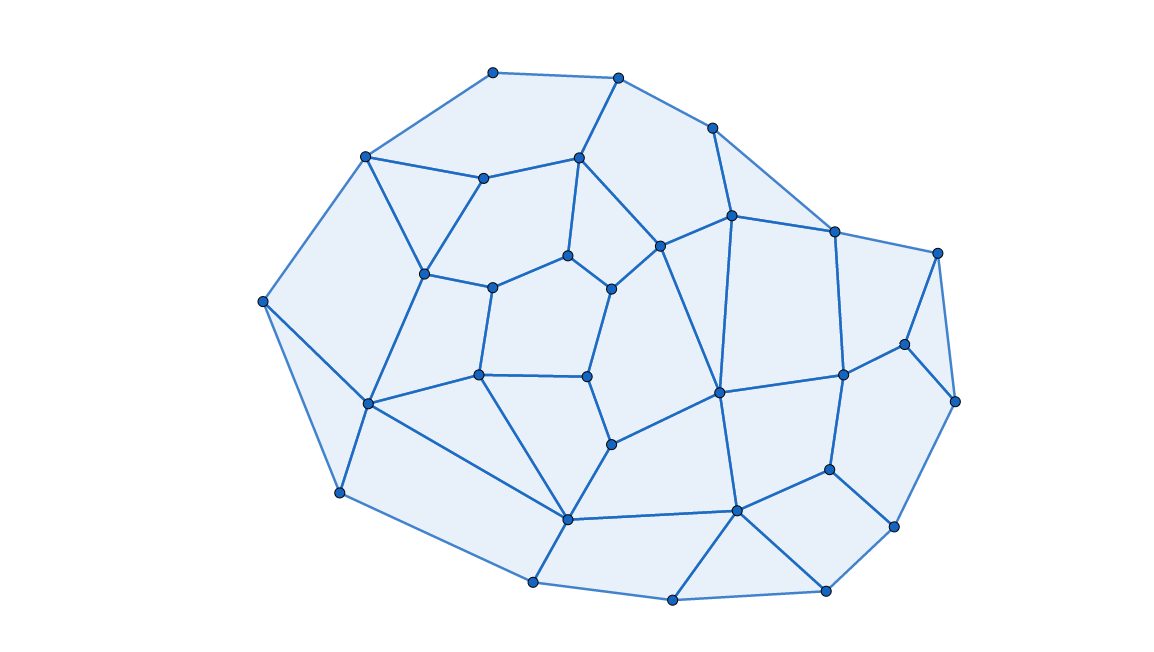}
\end{figure}
\end{minipage} \hfill
\begin{minipage}{0.45\textwidth}
An example of a polygonal complex, where each vertex and each face has a degree of at most 5. One may define on this complex a $2D-CLH^*(5,d)$ instance by assigning to each star and plaquette a Hamiltonian acting on the attached edges, where those Hamiltonians mutually commute.

\end{minipage}

{~}

We note that 
there is no restriction whatsoever on the topology of the complex $\k$; 
it can be of any genus, and may or may not include a boundary. 
We impose one condition on $\k$, which is a 
metric-geometric condition that we call quasi-Euclidity (though of similar flavor, it shouldn't be confused 
with the nearly-Euclidean condition of \cite{3local}). 
This condition ensures that the surface induced by the 
complex admits a triangulation in which  
the triangles may be slim (as in hyperbolic geometry) and 
may be fat (as in elliptic geometry) but only up to some constant. 
This makes the complex in some sense Euclidean up to a 
constant distortion, and prevents ``wild'' situations. 
Any physically natural 2D setting should be covered by this. 

Our main two results are:

\begin{customthm}{1}
\label{mainthm1}
The $2D-CLH^*(k,2)$ problem on quasi-Euclidean complexes 
is in NP. 
\end{customthm}
\begin{customthm}{2}
\label{mainthm2}
For any instance of $2D-CLH^*(k,2)$ defined on a quasi-Euclidean complex, 
there exists a polynomial
depth quantum circuit which prepares a groundstate.
\end{customthm}


Importantly, these results replace 
the mysterious feature of Schuch's result \cite{Schuch} 
providing a proof for containment in NP without an 
efficient groundstate 
description, by one in which the groundstate can be efficiently classically 
described; this seems to strengthen the common feeling 
that containment in NP should go hand in hand with 
efficient description for the groundstate.  
Moreover, our results hold for a wide class of cases, which includes not 
only the $4$-local case in a square lattice of Schuch \cite{Schuch}, but CLHs with arbitrary locality $k$, that are defined on any quasi-Euclidian 2D complex.
We remark that our definition of 
$2D-CLH^*(k,2)$ unfortunately does not capture the most general 
$k$-local quantum systems of qubits embedded on a surface (see Section \ref{sec:formulation} and Appendix \ref{apx:2d} for more details).

\subsection{Proof overview}

Our starting point is a folklore quantum algorithm for preparing 
the groundstates of the toric code.  
Recall that the toric code Hamiltonian \cite{ToricCode} 
acts on qubits set on the edges 
of an $n \times n$ grid with boundary conditions which make it topologically a torus. The Hamiltonian has two types of constraints, 
one for each vertex (star) denoted $s$, and one for each face (plaquette) denoted $p$: 
\begin{equation}
A_s=\bigotimes_{e \in s}\sigma_z^e,\qquad B_p=\bigotimes_{e\in p}\sigma_x^e. 
\label{eq:toric}
\end{equation}
$$H=-\sum_s A_s -\sum_p B_p$$
The groundstates of this Hamiltonian form a code space, 
and exhibit global-entanglement.

Consider creating ``holes'' in
the torus, by removing a small fraction of the plaquettes, 
in a regular manner. Figure \ref{fig:PuncturedToricCode} (A) shows
how by removing enough plaquettes we are left 
with a {\it punctured Hamiltonian} $\tilde{H}$, 
which involves {\it two local} interactions
between {\it super-particles} comprised each of
constantly many qubits. By \cite{BV} there is a
constant depth quantum circuit which prepares a groundstate (denote it $\ket{\psi}$) for $\tilde{H}$.

This doesn't seem at first as real progress, since
 $\ket{\psi}$ is a trivial state, whereas groundstates 
of the original Hamiltonian are globally entangled. 
The key idea is that now we can 
{\it correct for the plaquettes we have removed}, using  
the known idea of applying string operators connecting pairs of ``holes''.  

To do this, we first {\it measure} in the state $\ket{\psi}$ each of the 
plaquette terms which were removed. Due to the commutation
relations, the resulting state is {\it still} 
a groundstate of $\tilde{H}$ but 
now it is also an eigenstate of the toric
code, with  a known eigenvalue for each of the terms. 
Viewing the toric code as a subcode of the  
punctured code (the groundspace of the punctured Hamiltonian $\tilde{H}$), 
what we now need is 
a set of {\it logical operators} in the punctured code, 
that act within it and can 
transform our state into a toric code groundstate. 

To this end, we recall the notion of {\it string} operators which 
are Pauli operators 
acting on the paths (strings) connecting a pair of holes \cite{ToricCode}. 
Such an operator changes 
the values of the measurements corresponding to the constraints in 
both holes, while keeping all the other values intact. Notice 
that this process always works on {\it pairs} of holes.  The dependency
relations between the local terms 
($\prod_{s}A_{s}=\prod_{p}B_{p}=\identity$) \cite{ToricCode}
imply that for any eigenstate of the toric code there is  
an {\it even} number of plaquette (and also star) terms which are in their excited states.  Since all plaquettes in the punctured 
Hamiltonian are satisfied (i.e., not excited), it follows that 
there is an even number of excited plaquettes out of those 
which we removed, and thus such a pairing exists.

\begin{figure}
\begin{subfigure}{0.4\textwidth}
\caption{Punctured Hamiltonian}
\includegraphics[scale=0.23]{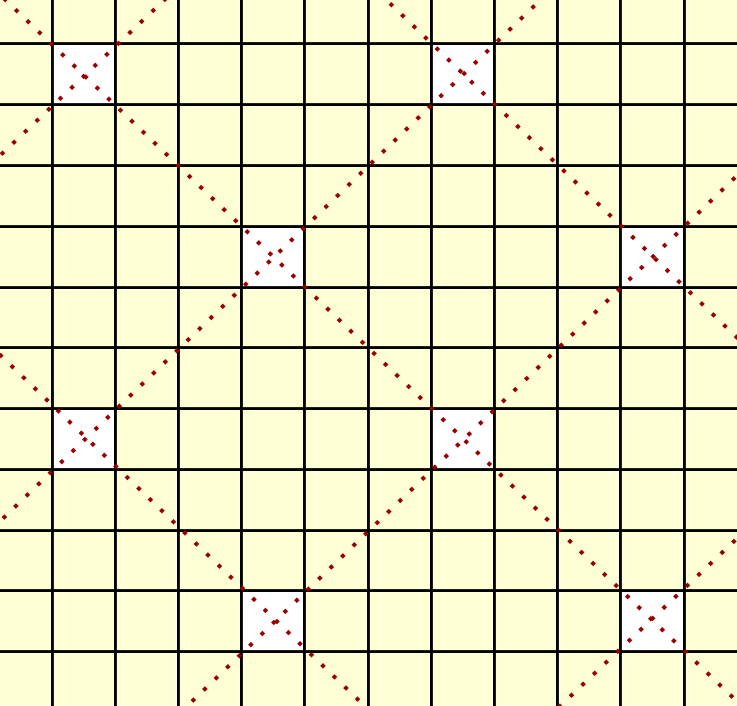}
\end{subfigure}
\begin{subfigure}{0.4\textwidth}
\caption{Logical Operators}
\includegraphics[scale=0.23]{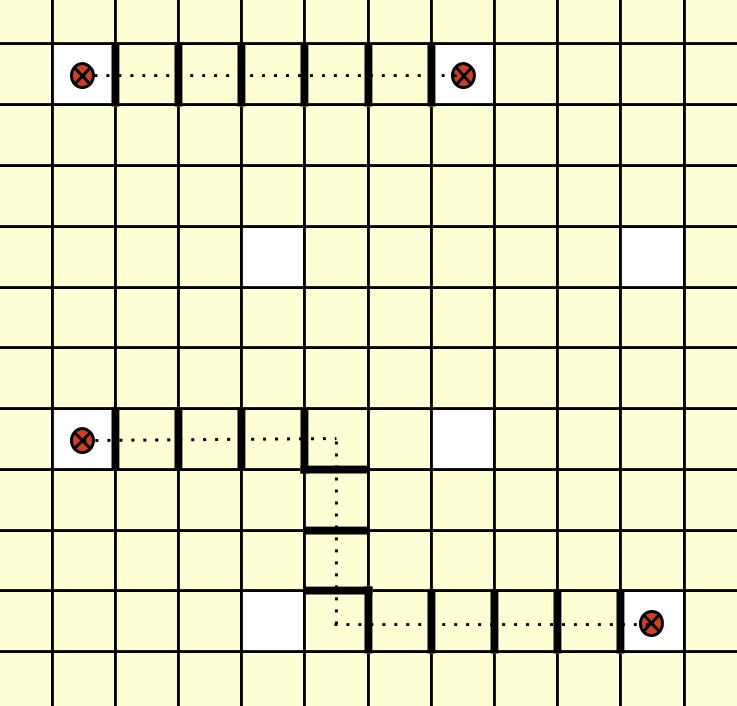}
\end{subfigure}
\caption{(A)  The white squares are the holes. The dotted lines induce a partition of the set of qubits (edges) to squares (tilted in 45 degrees), which are the super-particles, each 
containing a constant number of qubits. Every local term 
(star or plaquette) of the punctured Hamiltonian acts on qubits which belong to at most 2 super-particles.
(B): A hole with a spot inside indicates an excitation (i.e. a violation). 
The dotted lines are string logical operators (copaths) which annihilate particles in pairs. The edges in bold denote the qubits on which the
logical operator acts.}
\label{fig:PuncturedToricCode}
\end{figure}

 
Note that we could have actually removed {\it all} plaquettes, 
resulting in a punctured Hamiltonian $\tilde{H}$
consisting only of $A_s$ terms; Starting with the state $\ket{0^n}$, which is
a groundstate of $\tilde{H}$, we could then proceed as in  
the above algorithm, to derive a groundstate of the toric code 
(without any help of the prover).
We will make use of both approaches in this paper; 
the ``regular holes'' approach is the one we will generalize (conceptually) to more general instances, while the second more specific approach is used 
as a subroutine in our final algorithm, for technical reasons. We 
will thus present and prove it formally in Section \ref{sec:toric}.  
{~}

\subsubsection{Physical interpretation}\label{subsec:physical_int}
The toric code has a physical interpretation which 
will be very useful for us \cite{ToricCode}.  The value of the edges in the $\sigma_x$ and $\sigma_z$ basis are interpreted as a $\mathbb{Z}_2$ vector potential or electric field, respectively. When a constraint is violated, 
we interpret this as if an elementary excitation, or a particle, is created.  
The star constraints can be viewed as requiring that the 
electric flux from the vertex (namely the values of the qubits 
in the computational basis) is zero, i.e., that this 
vertex will have no electrical charge.
If a vertex constraint is violated, we say that there is an ``electric 
charge'' at that vertex. Likewise, the plaquette constraints require that the magnetic flux which  passes through the face is zero (mod $2$). If a plaquette constraint is violated we say that there is a "magnetric vortex" in this plaquette \cite{ToricCode}. 
The toric code consists of the states in which neither
electrical charges, nor magnetic vortices appear. The punctured system
however allows particles to be created at the sites which
we have removed. After measuring these terms, we know exactly 
where these particles are. It is left to annihilate them.  
Having a closed surface with no boundary, such as the torus, 
the total charge on it, as well as the 
total magnetic flux passing through it, must be zero (as Gauss and 
Stoke's laws imply, respectively). 
This means that there must be an even number of electrical charges, 
and an even number of magnetic vortexes,  
which can then be annihilated in pairs, by what is called 
``string operators'' connecting pairs of charges or pairs of vortexes (see \cite{ToricCode}). In the above algorithm for the toric code we only needed to   
annihilate magnetic vortices (plaquettes).

{~}
\subsubsection{From toric code to general $\sp$}
It is far from clear how the methods above concerning the toric code can 
be applied to general 2D CLH systems; after all, surface
codes seem to be an extremely restricted type of 2D CLHs (where the local terms must take the form of tensor
products of either $\sigma_x$ or $\sigma_z$ Pauli operators), whereas we are concerned with arbitrary commuting local terms. 
Theorem \ref{thm:equiv} in Section \ref{sec:equiv}  provides our first main step in the proof: we show that all $2D-CLH^*(k,2)$ instances are "equivalent to the toric code permitting boundaries". This in particular means that if 
all terms, stars and plaquettes, act non-trivially on all of their attached edges, (plus $\k$ is closed, i.e topologically has no boundary),  
then the instance is, up to a minor modification, equal to the toric code. 
In the general case, terms may act trivially on some of their qubits (edges); 
we will call such edges boundary/coboundary edges. Theorem \ref{thm:equiv} 
says that $\sp$ instance are virtually the toric code, except for those 
essentially 1D behaving boundary areas (and thus the term 
"permitting boundaries"). The proof of this structure theorem 
relies heavily on the C*-algebraic techniques mentioned earlier. 
We emphasize that Theorem \ref{thm:equiv} holds only after some transformation 
of the instance to one with no "classical qubits" whose value is simply a classical bit which can be provided by the prover 
(see subsection \ref{subsec:classical}).

\subsubsection{Constructing the Punctured Hamiltonian}
The above equivalence theorem raises 
the idea of using a similar algorithm as for the toric code groundstates, 
and somehow handling the special boundary/coboundary qubits. 
However, we encounter two challenges. 
First, we do not have sufficient control on operators 
near the boundary/coboundary. If we carelessly tear out holes 
in their vicinity,  
we might not know how to repair them- 
the correcting process of the toric code heavily relies on the specific commutation and anti-commutation relations between a string operator and the Hamiltonian terms (equation \ref{eq:toric}). We handle this difficulty by tearing out
holes only in the interior regions (that is regions 
without boundary/coboundary qubits) where we do have
resemblance to the toric code. It turns out that there is no 
need to tear holes close to boundary/coboundary qubits as 
in some sense these special qubits are already punctured: by definition such qubits are not  
surrounded by Hamiltonians acting on them non-trivially.

The second challenge is that we do no longer have the 
dependencies $\prod_s{A_s}=\prod_p{B_p}=\identity$ that ensured earlier 
an even number of excitations of any given type, and so the idea 
of fixing holes in pairs is irrelevant. In the physical interpretation, the latter means that the total charge on the manifold can be different than 0 since now flux can escape through the boundary. In section \ref{sec:puncture} we 
show that the curse of boundaries is 
in fact a blessing, since now we can also dump excitations to 
the boundary/coboundary with string operators, similarly to logical 
operators in surface codes \cite{BoundaryCodes} 
(figure \ref{fig:PuncturedSurfaceCode}).

The latter idea, which can be viewed as the main conceptual 
idea in the paper, introduces a new challenge - we have two types of 
special qubits. Boundary qubits give rise to copath string 
logical operators whereas coboundary qubits give rise to 
path string logical operators. We cannot expect that puncturing 
only plaquette terms out of the surface will allow us to fix them later on. 
Figure \ref{fig:gridLogicalOperators} shows simple examples of systems in which only one type of term (star/plaquette) have access to the boundary/coboundary via copath/path. In short, plaquettes play nicely with boundary edges whereas stars play nicely with coboundary edges.  

\begin{minipage}{0.6\textwidth}
\begin{figure}[H]
\caption{\label{fig:gridLogicalOperators} Logical operators}
\begin{subfigure}{0.45\textwidth}
\includegraphics[scale=0.20]{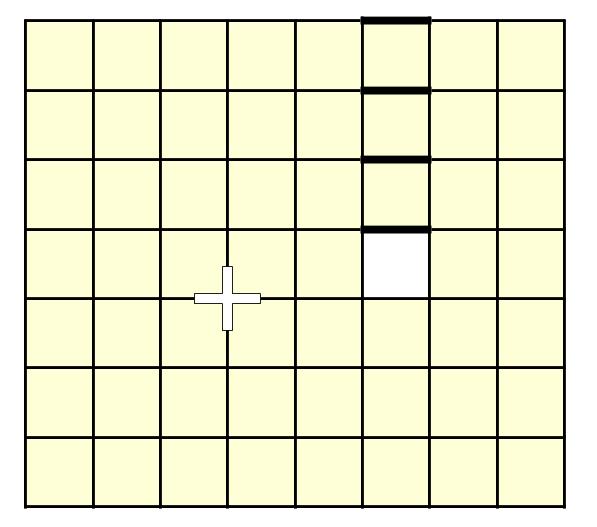}
\end{subfigure}
\begin{subfigure}{0.45\textwidth}
\includegraphics[scale=0.20]{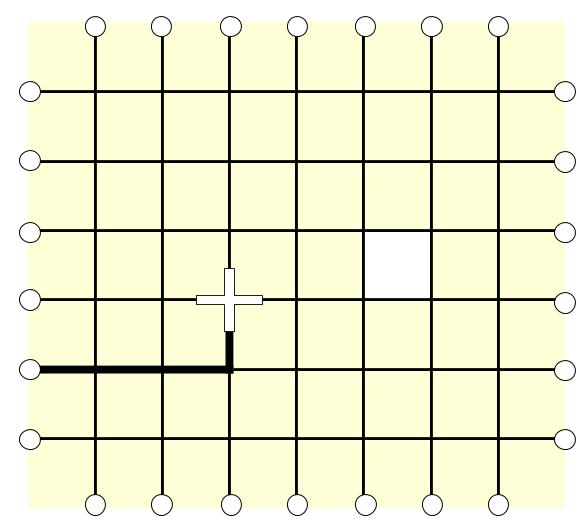}
\end{subfigure}
\end{figure}
\end{minipage} \hfill
\begin{minipage}{0.37\textwidth}

The white plaquette and the white plus indicate holes. 
In a complex with boundary but no coboundary only plaquette holes can be connected via a copath to utilize a logical operator, 
whereas in a complex with coboundary but no boundary only star holes 
can be connected via paths to utilize a logical operator.
\end{minipage}

{~}

A major technical effort in the paper is proving 
Lemma \ref{lem:either-path-or-copath} which roughly states that for any 
adjacent plaquette and star, at least one of them has access to the 
boundary/coboundary (unless they are both already touching the 
boundary/coboundary), hence a hole in one of them will be fixable 

With this in mind, we construct the punctured Hamiltonian as follows: 
we start by considering the set $\mathcal{W}$ of "fixable" terms. These are terms which are not in the boundary of the system (and thus are in the form of a toric code term) and in addition have access to the boundary or coboundary via a copath or path depending on whether it is a plaquette or star term respectively (see Definition \ref{def:access} and Figure \ref{fig:ribbon}). By Lemma \ref{lem:either-path-or-copath} the fixable holes are very ``dense''. 
We shall not hesitate to remove all of those terms since, by how the elements of the set $\mathcal{W}$ were chosen,  we can correct their values later on. 

We call the Hamiltonian obtained by removing all of the terms in $\mathcal{W}$ the {\it punctured Hamiltonian} $\tilde{H}$.

\subsubsection{2-locality of the punctured Hamiltonian}

Lemma \ref{lem:either-path-or-copath} guarantees that at any large enough constant size area, either there are boundary qubits (recall these are qubits which are acted trivially by at least one of its surrounding terms) which may serve as a hole, or else there must be a fixable term in that area, i.e a member of $\mathcal{W}$, which was removed. 
In the case of the grid it is now very simple to generate a $2$-local 
structure among constant size super-particles: 
just consider a coarse grained grid of $5\times5$, and use Lemma \ref{lem:either-path-or-copath} to conclude that there must be some hole inside 
each $5\times5$ square. However we are allowing much more general geometries than the grid; it is here and only here, that we make use of the quasi-Euclidity condition. This is what allows us to follow a similar process, and 
to tear holes in some regular manner. Technically, we need to apply Moore's bound (Fact \ref{fact:moore}) to bound the number of edges (qubits) 
which belong to any super-particle resulting from the process; 
together some other combinatorial arguments the proof goes through.

\begin{figure}

\begin{subfigure}{0.4\textwidth}
\caption{Punctured Hamiltonian}
\includegraphics[scale=0.23]{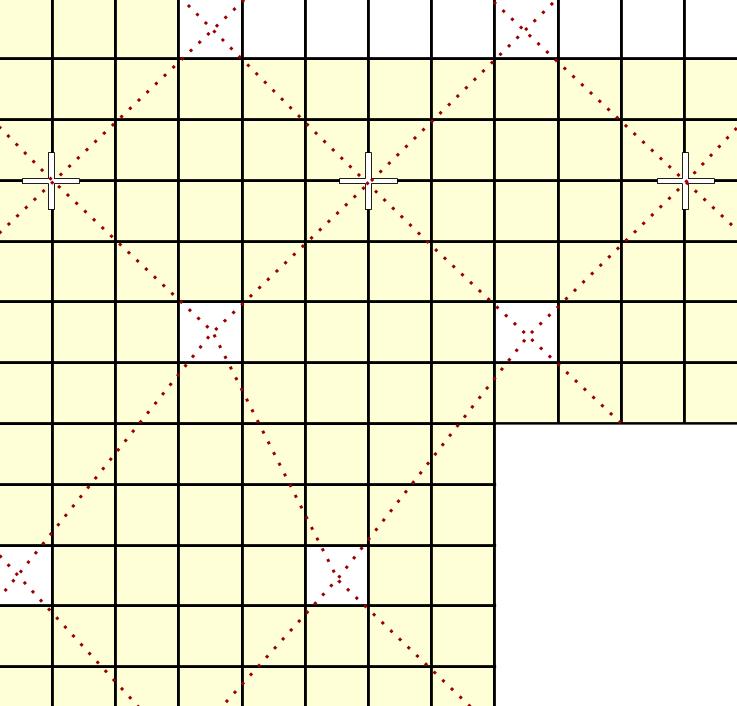}
\end{subfigure}
\begin{subfigure}{0.4\textwidth}
\caption{Logical Operators}
\includegraphics[scale=0.23]{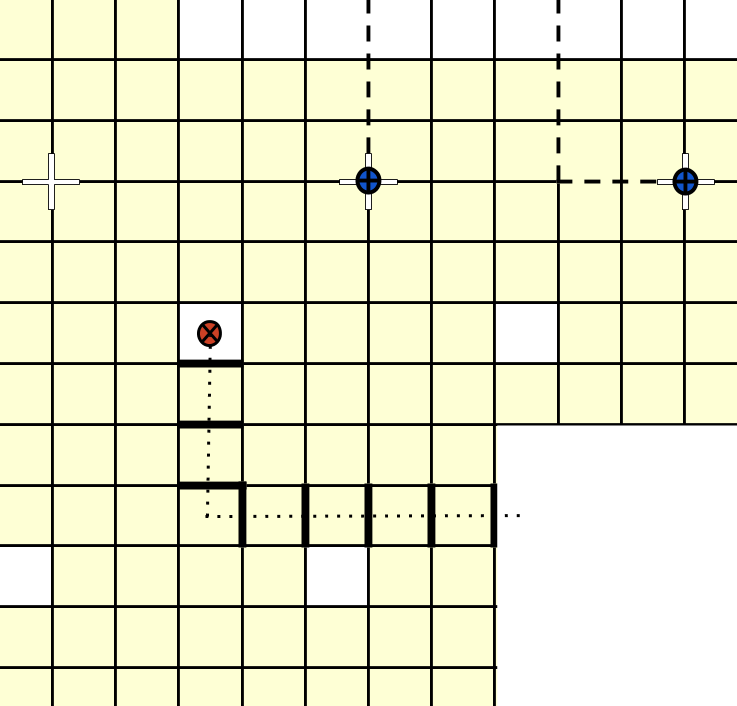}
\end{subfigure}
\caption{
(A) Even when boundaries/coboundaries exist, one can tear out holes to obtain a 2-local instance w.r.t superparticles of constant size.
(B) After measuring each hole, it remains to correct it if needed by connecting it to the boundary/coboundary via a string operator depending on the hole type (i.e plaquette/star).
\label{fig:PuncturedSurfaceCode}
}
\end{figure}

Now that the punctured Hamiltonian is 2-local, we again are guaranteed that
a groundstate can be generated by a constant depth quantum circuit \cite{BV}.
This is the only place where the prover is needed. Note that 
this groundstate is in general not the groundstate
of the original Hamiltonian, yet, the fact that
we have torn out only terms of $\mathcal{W}$, namely the fixable terms, 
implies that we can apply the
approach of measuring them and correcting them with string operators to the boundary/coboundary of the system (Figure \ref{fig:PuncturedSurfaceCode} (B)).

\subsection{Organization of the paper}
In Section \ref{sec:formulation} we formalize the problem. Section \ref{sec:background} gives some background: "the induced algebra", "classical qubits", and notations. 
Section \ref{sec:toric} provides the efficient algorithm for generating toric code states which we use as a subroutine. Section \ref{sec:equiv} contains
Theorem \ref{thm:equiv}, stating that $\sp$ instances are "equivalent to the toric code permitting boundaries". Based on this, in 
Section \ref{sec:puncture} we prove lemma \ref{lem:either-path-or-copath} which shows that many fixable terms (those with "access to the boundary") exist, 
and define the punctured Hamiltonian, in which all these terms are removed. In Section \ref{sec:2local} we show that
the punctured Hamiltonian is indeed 2-local with respect to super-particles of constant size. 
Section \ref{sec:completingproof} combines all these results to prove Theorems 1,2. In Section \ref{sec:discussion} we discuss the results, their implications, and state 
open questions.

\section{Formulation of the problem \label{sec:formulation}}
\subsection{Definitions}
\begin{defn} [CLH instance]
	\label{def:CLH} An instance of  
	$CLH(k,d)$ consists of a set of Hamiltonian terms (Hermitian matrices) 
	acting on $n$ qudits (particles of dimension $d$), where each term acts non-trivially 
	on at most $k$ of the $n$ qudits. The norm of each term is bounded by $1$, 
	and the terms mutually commute.
\end{defn}

To be precise, we note that as usual, the Hermitian matrices are given with 
entries represented by poly(n) bits.

We consider the cases where the CLH instance is defined on a 2D
complex. The type of complexes we allow (see definition bellow) is a generalization of a simplicial
2-complex; while in simplicial complexes the 2-cells must be 2-simplexes (triangles), 
we allow the 2-cells to be any simple polygon. Topologically speaking, we may define a simple polygon to be any set homeomorphic to the closed disk $D=\{\textbf{x}\in \mathbb{R}^2 \mid ||\textbf{x}||\leq 1\}$ with some choice of a finite amount (at least three) of points on its boundary to be called the vertices of the polygon. The arcs on the boundary which connect two adjacent vertices are called the sides of the polygon. Such complexes are often called {\it polygonal
complexes} \cite{Graph2}.

\begin{defn}[polygonal complex]
\label{def:complex}A polygonal complex $\mathcal{K}$ is a collection of points (called 0-cells or vertices),
line segments (1-cells, or edges), and simple polygons (2-cells, or faces) glued to each other such that:
\begin{enumerate}
\item Any side of a 2-cell in $\mathcal{K}$ is a 1-cell in $\mathcal{K}$. Every endpoint of a 1-cell in $\k$ is a 0-cell in $\k$.
\item The intersection of any two distinct 2-cells of $\mathcal{K}$ is either empty or else it is a single 1-cell (along with its endpoints). The intersection of any two distinct 1-cells of $\k$ is either empty or else it is a single 0-cell.
\end{enumerate}
If all polygons have exactly three vertices then $\k$ is called a simplicial 2-complex. 
The 1-skeleton of $\k$ is by definition the graph obtained
by removing all 2-cells from $\k$. Finally, $\k$ is called two dimensional (2D)
if the topological space which it defines $\s=\bigcup \k$  is a surface.
\end{defn}

By surface we mean the topological definition of a surface\footnote{In many texts (e.g \cite{Hatcher}) second countability and Hausdorff are required in the definition as well. In our case however, we are only considering finite polygonal complexes which always satisfy these two conditions.} allowing boundaries \cite{surfaces}; that is a topological space such that each point in the interior has a neighborhood homeomorphic to $\mathbb{R}^2$ whereas each point in the boundary has a neighborhood which is homeomorphic to the the upper plane $\{(x,y)\in \mathbb{R}^2\mid y\geq 0\}$. We shall remark that if $\k$ is finite (which will be the only case
we consider) then $\s$ is compact. If in addition $\s$ has no boundary (in the ordinary topological sense) then we say that $\s$ (and thus also $\k$) is closed.

Note that 2D polygonal complexes have the property that every 1-cell is the face of at most two 2-cells (one if that 1-cell is in the boundary, and two if it is in the interior). That is because if 3 or more 2-cells are attached at that 1-cell then the neighborhoods of points in the interior of that 1-cell are neither homeomorphic to  $\mathbb{R}^2$ nor to the upper plane.

The 1-skeleton of $\k$ admit the natural graph metric in which the distance between any two vertices is the length of the minimal path between them, where the length of every edge is 1.

\begin{defn} [triangulation]
A \emph{triangulation} of a topological space $X$ is a finite simplicial 2-complex $\mathcal{T}$ together with a homeomorphism $f:\mathcal{T}\rightarrow X$. The 2-cells of $\mathcal{T}$ are called the \emph{triangles} of the triangulation. 
\end{defn}

The following definition is inspired by metric geometry in which hyperbolic spaces are roughly defined to be metric spaces which have only $r$-slim triangles - triangles which do not contain any ball of radius $r$; whereas elliptic metric spaces are such which have a bound on the {\it diameter} of triangles \cite{metric}.

\begin{defn} [quasi-euclidean 2D complex]
	\label{def:qe}
	Let $\k$ be a 2D polygonal complex with underlying surface $S$. A triangulation of $\s$ is said to be $(r,R)-$quasi-Euclidean for some $0<r<R$ if each of its triangles contains a ball of radius $r$ in $\k$ (w.r.t metric defined above) and the subgraph in it is of diameter at most $R$. The degree of a triangulation is by definition the maximal degree of its 1-skeleton.
	In the case where $\s$ admits such a triangulation we say that $\k$ is $(r,R)$-quasi-Euclidean.
\end{defn}

We emphasize that there is no demand from the triangulation to be
in any sort in accordance with the complex structure of $\k$  
(e.g vertices of $\mathcal{T}$ do not need to be located on vertices of $\k$). 

\begin{minipage}{0.5\textwidth}
\begin{figure}[H]	
\caption{\label{fig:qeComplex} Quasi-Euclidean polygonal complex}
\includegraphics[scale=0.35]{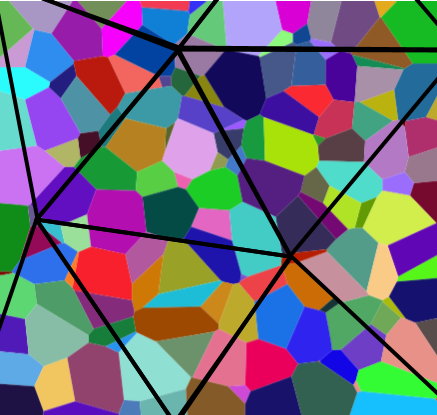}
\end{figure}
\end{minipage} 
\begin{minipage}{0.473\textwidth}
\vspace{4mm}
A triangulation $\mathcal{T}$ (dark lines) of the surface $\s$ on which the complex $\k$ lies. $\mathcal{T}$ is $(r,R)$-quasi-Euclidean with $r=2$,$R=12$ since each triangle contains a ball of radius $2$ but its diameter is less than $12$. The makes $\k$ a $(r,R)$-quasi-Euclidean complex.
Having each triangle contain a ball of radius $r\geq 2k$ (here $k=7$)
ensures that there exists a polygon which is contained in the triangle, as well as all other polygons touching it. The fact that the diameter of each triangle is at most $R$ implies that the number of edges in each triangle is bounded by a number dependent only on $R$ and $k$, by Moore's bound \cite{Graph2}.
\end{minipage}

\vspace{3mm}
\begin{defn} [$2D-CLH^*(k,d)$ instance]
\label{def:sp-CLH} Consider instances $x$ of $CLH(k,d)$ for which:

\begin{enumerate}
\item There exists a two dimensional polygonal
complex $\mathcal{K}.$
\item There exists a 1-1 mapping between qudits of $x$ and edges of $\mathcal{K}$.
\item There exists a 1-1 mapping between local terms of $x$ and the set
of vertices and faces of $\mathcal{K}$. 
\item If $h$ corresponds to a vertex $v$ then the set of qudits $\{q_{1},...,q_{r}\}$
which $h$ acts on corresponds to the set of edges $\{e_{1},...,e_{r}\}$
attached to $v$.
\item If $h$ corresponds to a face $f$ then the set of qudits $\{q_{1},...,q_{r}\}$
which $h$ acts on corresponds to the set of edges $\{e_{1},...,e_{r}\}$
which are in the boundary of $f$.
\end{enumerate}
We consider the restriction of this class to quasi-Euclidean complexes - those which admit a $(r,R)$-quasi-Euclidean triangulation of degree $D$, for some arbitrary constants $D>0$ 
and $R>r>2k$. We call such $2D-CLH^*(k,d)$ instances {\it quasi-Euclidean}.
\end{defn}

The quasi-Euclidean condition doesn't limit the topology in any way. Specifically, for any compact surface $\s$ there exists a quasi-Euclidean polygonal complex $\k$  such that $\s$ is its underlying surface (i.e $\s=\bigcup\k$) \cite{surfaces}.  This condition is needed only in Section \ref{sec:2local}. Hence in the following we ignore it and treat general $2D-CLH^*(k,2)$ instances; only in Section \ref{sec:2local} we will mention this condition again.

Another possible way to define a CLH on a 2D polygonal complex
is to place the qudits on the {\it vertices} rather than the edges, and then local terms are associated
with faces alone. We denote the class of such instances by $2D-CLH(k,d)$ (i.e without the star symbol) - the definition goes on the same line as Definition \ref{def:sp-CLH} though it is presented formally in the appendix - definition \ref{def:2D-CLH}.

The second definition captures the notion of a 2D system in a more general way. However, if our results can be generalized to $2D-CLH^*(k,d)$ for arbitrary $d$, this will in fact imply that they also hold for  $2D-CLH(k,d)$, under a mild condition similar to quasi-Euclidity (see Appendix \ref{apx:2d}).

To each of those classes corresponds the local Hamiltonian problem of deciding, given $a<b$ with $b-a<\frac{1}{poly(n)}$, whether
the ground energy of the system (i.e the sum of all local terms) is bellow $a$ or above $b$, provided the promise that one of these cases hold.
We use the same notation to denote both the class of such instances
(as in Theorem 2) and the corresponding decision problem (as in Theorem 1).

\section{Notation and Background}
\label{sec:background}
\subsection{Notations}
\label{subsec:notation}
Throughout this paper we use $\mathcal{H}$ to denote Hilbert spaces,
$q$ to denote qubits, and accordingly $\mathcal{H}_{q}$ to denote
the Hilbert space associated with the qubit $q$. $\mathcal{K}$ denotes
the complex on which the $\sp$ is defined whereas $\s$ denotes its underlying surface. We use $s$ to denote stars,
$p$ to denote plaquettes and let $|s|$ and $|p|$ denote the degree
of a star or a plaquette, i.e the number of edges which belong to $s$ or to $p$. $A_{s}$
denotes the local term which corresponds to $s$ and $B_{p}$ denotes
the local term which corresponds to $p$. $h$ denotes a local
term in general. We say that two stars (plaquettes) are adjacent if they share an edge,
and say that a star and plaquette are adjacent if they share two edges
(which is the only way a star and a plaquette can intersect). When more geometrical aspects are discussed we will
consider vertices instead of stars denoted by $v$, edges instead
of qubits denoted by $e$ and faces instead of plaquettes denoted
by $f$. We let $H$ denote
the sum of all local terms $H=\sum_{s}A_{s}+\sum_{p}B_{p}$ where
$s$ and $p$ range over the stars and plaquettes of the instance. When we construct a punctured Hamiltonian, i.e a Hamiltonian obtained by removing some terms from the original one, we will always denote it by $\tilde{H}$.

\subsection{The induced algebra}
\begin{defn}[induced algebra]
\label{def:induced-algebra}Let $h$ be an operator on a tensor product
Hilbert space $\mathcal{H}_{q_{1}}\otimes\mathcal{H}_{q_{2}}$ and
let $h=\sum_{i=1}^{m}h_{q_{1}}^{i}\otimes h_{q_{2}}^{i}$ be a Schmidt
decomposition\footnote{that is to say: $h_{q_{1}}^{i}\in\mathcal{L}\left(\mathcal{H}_{q_{1}}\right)$,
$h_{q_{2}}^{i}\in\mathcal{L}\left(\mathcal{H}_{q_{2}}\right)$  for each $i$ and the that sets $\left\{h^i_{q_1}\right\}^m_{i=1}$, $\left\{h^i_{q_2}\right\}^m_{i=1}$ are orthogonal with respect to the Hilbert-Schmidt inner product i.e $tr({h^i_{q_l}}^\dagger\cdot{h^j_{q_l}})=0$ for any $i\ne j$ and $l=1,2$)} of $h$.
The induced algebra of $h$ on $\mathcal{H}_{q_{1}}$ is denote by
$\mathcal{A}_{\mathcal{H}_{q_{1}}}^{h}$ or in short $\mathcal{A}_{q_{1}}^{h}$
and is defined to be the C*-algebra generated by $\{ I \} \cup \{ h^i_{q_1}\}^m_{i=1} $ ($I$ denotes the identity operator).
\end{defn}

\subsection{Classical qubits}
\label{subsec:classical}

The equivalence to the toric code which we are aiming for can be shown only after performing a certain reduction of removing "classical qubits". 
Classical qubits are classical in the sense that they do not participate
in the entanglement of the system and consequently, the prover may hand us its correct value as a classical bit. 

\begin{defn} [trivial qubit]
\label{def:classical}A qubit (or qudit) is called trivial, if no
local term acts on it non-trivially. 
\end{defn} 

\begin{defn} [classical qubit] 
A qubit (or qudit) is called classical if its
Hilbert space can be decomposed into a direct sum of 1-dimensional
subspaces which are invariant under all local terms in the Hamiltonian $H$.
\end{defn}

When we say that a Hamiltonian $h$ acts trivially on a certain qubit we simply mean that it can be written as $h=I\otimes h^\prime$ where $I$ is the identity operator on that qubit, and $h'$ acts only on other qubits.

Note that due to the low dimension of qubits, once such a non-trivial
direct sum decomposition exists then the subspaces must be one dimensional
and so the qubit is classical. Note also that every trivial qubit
is in particular classical - any direct sum decomposition will do. The following claim says that whenever there is a classical qubit $q$, the instance can be reduced to a new instance in which it is a trivial qubit.

\begin{claim} [removing classical qubits]
\label{claim:no-classical-qubits} To derive theorems 1,2 it is sufficient to prove it under the restriction of $2D-CLH^*(k,2)$ to instances with the condition that every classical qubit is trivial.
\end{claim}

\begin{proof}
Appendix \ref{apx:classicalQubits}	
\end{proof}

Thus, we shall assume from now on that all classical qubits were turned to be trivial qubits.

\section{Generating a toric code state}
\label{sec:toric}
The toric code is a special case of a $\sp$ instance. We shall not restrict to the particular setting of a grid on a torus, so by saying toric code we refer to any $\sp$ instance defined on a closed complex $\k$ (i.e it topologically has no boundary) with the usual star and plaquette local terms (equation \ref{eq:toric}). 

Starting with the state $\ket{0}^{\otimes n}$, we measure {\it all} plaquettes and record the measurement results by $\bar{\lambda} =  (\lambda_p)_p$ (where $\lambda_p=\pm 1$). As a result, the system collapses to a state corresponding to the measured values:
$\ket{\psi_{\bar{\lambda}}}$.  Note that $\ket{\psi_{\bar{\lambda}}}$ is a toric code state (i.e a groundstate of the Hamiltonian given in equation \ref{eq:toric}) precisely when  $\lambda_p = 1$ for each plaquette $p$. 

Whenever we have two plaquettes $p_1$,$p_2$ with $\lambda_{p_1}=\lambda_{p_2}=-1$ we can connect them by a copath $\gamma^*$, apply $L^*=\bigotimes_{e\in\gamma^*}Z_e$, and obtain a new state $\ket{\psi_{\bar{\lambda^\prime}}}$ where $\lambda$ and $\lambda^\prime$ are the same except for the value on the plaquettes $p_1$,$p_2$ (see Appendix \ref{apx:logical}). In other words, a pair of plaquette terms which are in their excited state can always be relaxed. After matching pairs of excitations, and annihilating them by applying string operators between
them, we obtain a toric code state.  It is thus left to show that such a matching always exists:

\begin{claim} [even amount of excitations]
The number of plaquettes $p$ for which $\lambda_p=-1$ is even.
\end{claim}

\begin{proof}
Since $\k$ is closed so $\prod_p B_p = \identity$ (and also $\prod_s A_s = \identity$). Therefore:
$$\ket{\psi_{\bar{\lambda}}}=\identity\ket{\psi_{\bar{\lambda}}}=
\prod_{p}B_{p}\ket{\psi_{\bar{\lambda}}}=\prod_{p}\lambda_{p}\ket{\psi_{\bar{\lambda}}}=(\prod_{p}\lambda_{p})\ket{\psi_{\bar{\lambda}}}$$ It follows that $\prod_{p}\lambda_{p}=1$.
\end{proof}

This is summarized by the following algorithm.

\subsubsection{Algorithm - constructing a toric code state (folklore):}
\label{alg:toric} 
\begin{enumerate}
\item Start with the tensor product state $\ket 0^{\otimes n}$.
\item For each star $p$ measure $B_{p}$ and record the measured value
$\lambda_{p}$.
\item As long as $-1\in\left\{ \lambda_{p}\right\} _{p}$ choose two stars
$p_{1},p_{2}$ with $\lambda_{p_{1}}=\lambda_{p_{2}}=-1$, find a
copath $\gamma^*$ connecting them (with some linear time path-finding classical
algorithm) and apply $Z$ along that copath, that is the operator $L=\bigotimes_{q\in\gamma} Z_q$. Then change the values of $\lambda_{p_{1}},\lambda_{p_{2}}$
from $-1$ to $1$. 
\end{enumerate}

It is not hard to be convinced that a similar approach works also for a variation of the toric code where each term is as in the toric code but with some scalar factor - we explain this in Appendix \ref{apx:Defected-toric-code}. This remark is relevant since the equivalence to the toric code (which we formulate in the following section) allows such factors.

\section{Equivalence to the toric code} 
\label{sec:equiv}
We now formulate the notion of equivalence between general $2D-CLH^*(k,d)$ instances and the toric code. 

\begin{defn} [boundary/coboundary qubit]
\label{def:boundary}A qubit is said to be in the \emph{boundary} of the
system if it is acted non-trivially by at most one plaquette; it
said to be in the \emph{coboundary} of the system if it is acted non-trivially by at most one star. Other qubits are said to be in the \emph{interior}. A local term which acts only on interior qubits is said to be in the \emph{interior of the system}.
\end{defn}

Qubits that live on edges which are topologically on the boundary of 
the manifold are of course in the boundary of the system;  however 
qubits which are (topologically) in the interior of the manifold can also 
be in the boundary/coboundary of the {\it system} if 
a Hamiltonian term {\it acts trivially} on them.
When this happens, these qubits serve, in spirit, as ``holes''. We will later exploit this fact in order to tear out holes only in the {\it interior} of the system to obtain the 2-local punctured Hamiltonian and a constant depth circuit that generates groundstate for it. 

Following \cite{BV}, we will make use of the notion of induced algebras (Definition \ref{def:induced-algebra}) of any term in the Hamiltonian, on any set of qubits it acts on. 
The induced algebra from a star (plaquette) term $s$ ($p$) on qubits 
$q_1,...q_r$ is denoted 
$\mathcal{A}_{q_{1},...,q_{r}}^{s}$ ($\mathcal{A}_{q_{1},...,q_{r}}^{p}$).   
We can now state the definition of 
equivalence to the toric code:

\begin{defn} [equivalence to the toric code permitting boundaries]
\label{def:equiv}
An instance of $2D-CLH^*(k,2)$ is said to be {\em equivalent
to the toric code} if its underlying surface   
$\mathcal{S}$ is closed (it topologically doesn't have boundary) and there
exists a choice of basis for each qubit such that 
$A_{s}\in\left\langle Z^{\otimes|s|}\right\rangle \backslash{} \mathbb{C}\cdot I$,
$B_{p}\in\left\langle X^{\otimes|p|}\right\rangle \backslash{} \mathbb{C}\cdot I $
for any $s,p$. 

An instance is said to be {\em equivalent to the toric code permitting boundaries} if there exists a choice of basis for each qubit such that:

\begin{enumerate}
\item $\mathcal{A}_{q_{1},...,q_{r}}^{s}
=\left\langle Z^{\otimes r}\right\rangle $
for any star $s$, for $(q_{1},...,q_{r})$ a copath
of qubits of $s$ which are not in the coboundary, with no two consecutive 
qubits in the boundary. 
\item $\mathcal{A}_{q_{1}^{\prime},...,q_{r}^{\prime}}^{p}=\left\langle X^{\otimes r}\right\rangle $
for any plaquette $p$, for $(q_{1}^{\prime},...,q{}_{r}^{\prime})$
a path of qubits of $p$ which are not in the boundary, with no two 
consecutive qubits in the coboundary.
\end{enumerate}

\end{defn}

\begin{thm} [equivalence to the toric code permitting boundaries]
\label{thm:equiv}
Every $2D-CLH^*(k,2)$ instance (after removing all classical qubits as described in subsection \ref{subsec:classical}) is equivalent
to the toric code permitting boundaries. In particular, if it has
no qubits which are in the boundary or in the coboundary then it is equivalent to the toric
code. 

\end{thm}

The proof of this theorem is in Appendix \ref{apx:equiv}. 
It is based on a classification of the possible induced algebras of a
Hamiltonian on a single qubit in the interior (Lemma \ref{lemma:qubit-algebra}, 
 \ref{lem:schurs-lemma}) which shows that these algebras 
are always generated by a {\it single}
 Pauli operator (i.e., an operator which 
is equal to a Pauli matrix up to a change of basis). 
Moreover, the main technical part is 
captured by Lemma \ref{lem:strong-commutation} which provides a 
severe restriction on the induced algebras on {\it pairs} of qubits
(which are in the interior, roughly), essentially  
showing that they must be similar to those of the toric code. 
This analysis involves a close and fairly technical
study of the implication of the commutation relations between 
the Hamiltonians on the algebras that they induce. 

An immediate implication of Theorem \ref{thm:equiv} is that we now know how to generate a groundstate for any $2D-CLH^*(k,2)$ instance which has no qubits in the boundary or coboundary of the system, since such instances are equivalent to the toric code (see Appendix \ref{apx:Defected-toric-code} for a more detailed explanation).

\section{Construction of punctured Hamiltonian}  \label{sec:puncture}

We are now ready to show how we can generate a groundstate of an arbitrary quasi-Euclidean $2D-CLH^*(k,2)$  instance, even when there are qubits in the boundary/coboundary.

\begin{defn} [access to the boundary/coboundary]
\label{def:access} 
A star $s$ is said to have \emph{access to the coboundary} if there
exists a path $\gamma$ starting from $s$ which ends at a coboundary
edge such that $L=\bigotimes_{q\in\gamma}X_{q}$ anti-commutes with
$A_{s}$ and commutes with any other local term. Similarly, a plaquette
$p$ is said to have \emph{access to the boundary} if there exists
a copath $\gamma^{*}$ starting from $p$ which ends at a boundary edge such
that $L^{*}=\bigotimes_{q\in\gamma^{*}}Z_{q}$ anti-commutes with
$B_{p}$ and commutes with any other local term.
\end{defn}

Access to the boundary or coboundary means that either $L^*$ or $L$ serve as an appropriate logical operator for the corresponding plaquette or star respectively (see Appendix \ref{apx:logical} for more about logical operators in surface codes).

\begin{lem} [{\bf Maim lemma:} access to the boundary/coboundary]
\label{lem:either-path-or-copath} 
Let $s,p$ be adjacent star and
plaquette which are in the interior of the system. Then either $s$ has access to
the coboundary, or $p$ has access to the boundary.
\end{lem}

\begin{proof}
Appendix \ref{apx:proof_of_lem:either-path-or-copath}
\end{proof}

The proof of Lemma \ref{lem:either-path-or-copath} relies on the a further study of the induced algebras near the boundary/coboundary of the system. The idea is to start with an edge shared by $s$ and $p$ and start drawing a ribbon from it which is briefly a juxtaposition of a path and an adjacent copath (see Definition \ref{def:ribbon} and Figure \ref{fig:ribbon}). We do this until we encounter a boundary/coboundary edge. At areas far from the boundary, we are in a regime which look like the toric code and thus the desired commutation and anti-commutation relations hold. It is then corollary \ref{cor:properboundary2} that provides restrictions on the induced algebras near the boundary/coboundary qubit which in turn implies that either the path or the copath within the ribbon can serve as the support for a logical operator which can correct $p$ or $s$ respectively.

\noindent {\bf Construction of punctured Hamiltonian}:
Let $\mathcal{W}$ denote the set consisting of all stars and plaquettes in the interior of the system which have
access to the coboundary or to the
boundary, respectively. This set can be thought of as the set of ``fixable'' terms.
Let $\tilde{H}$ be the \emph{punctured Hamiltonian:} the local Hamiltonian
obtained by replacing all terms which are in $\mathcal{W}$ by the identity operator.  

\section{2-locality of the punctured Hamiltonian} 
\label{sec:2local}
We now show with the help of Lemma \ref{lem:either-path-or-copath} that the punctured Hamiltonian $\tilde{H}$ has so many holes that it is 2-local. 

The division to superparticles is based on the quasi-Euclidean condition (this is the only place we use this condition). 
Recall that by definition, the quasi-Euclidity condition (Definition \ref{def:qe})
provides us with a triangulation $\mathcal{T}$ of $\s$ of degree
$D=\mathcal{O}(1)$ such that each triangle contains a ball of radius
$2k$  and is of diameter $R=\mathcal{O}(1)$ (with the ordinary graph metric with edge length 1). 

We now construct a graph which will help us divide the qubits to 
superparticles. The vertices of this graph will be associated with terms in 
the Hamiltonian. A local term can be associated with a point
in the surface in a natural way: each star is naturally realized 
as the vertex which is associated with it, and each plaquette $p$
is associated with some arbitrarily chosen point in its interior 
to be called ``the center of the plaquette''. This allows us to 
precisely speak of a local term as a point on the surface.

\begin{claim} [punctured triangles]
\label{claim:triangle-center}For each triangle $T\in\mathcal{T}$
there exists a term $h$ of $\tilde{H}$ 
such that all of the edges attached to it (i.e the edges associated with the qubits which $h$ acts on), are fully contained in $T$ 
and moreover, $h$ acts trivially on at least one of its qubits.
\end{claim}

\begin{proof}
Appendix \ref{apx:2-local}
\end{proof}

Choose such a term $h$ for every triangle $T\in\mathcal{T}$ and
call it ``the center of the triangle $T$''. Such a term acts trivially on some 
edge $e$ (when considered as a term in $\tilde{H}$; 
if $h$ was removed, then this term is in fact the identity). 
In addition, for each 1-cell of $\mathcal{T}$, 
that is a side of a triangle $T\in\mathcal{T}$, choose some point
in its interior to be called ``the center of the 1-cell''. Then
connect each triangle center with 3 paths to the centers of the sides
of $T$. Those paths should be non intersecting, contained in the
interior of $T$ (except at the end of the paths) and in addition
must satisfy one more condition: clearly, those three non-intersecting
paths divide $T$ into 3 regions; the paths should be drawn such that
$e$ belongs to one region and all the other edges of $h$ belong
to the two other regions (that way $h$ will act non-trivially on at most 2 regions). To be sure that such paths can always be drawn, 
it suffices to show it for an equilateral triangle - this can of course be 
done. Then the general case is obtained as a 
homeomorphism of the triangle (see figure \ref{fig:puncturedTriangle}).

\begin{minipage}{0.5\textwidth}
\begin{figure}[H]	
\caption{\label{fig:puncturedTriangle} Choosing a holes in each triangle and 	 separating to regions}
\includegraphics[scale=0.26]{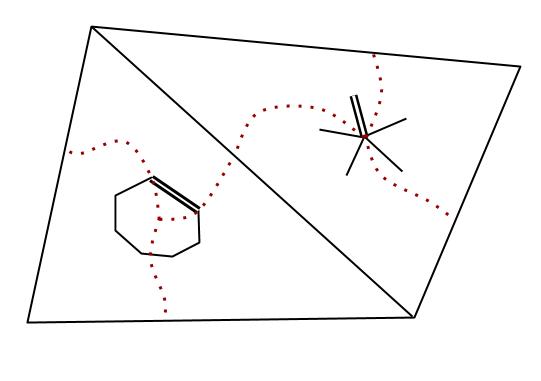}
\end{figure}
\end{minipage} \hfill
\begin{minipage}{0.45\textwidth}
According to Claim \ref{claim:triangle-center}, every triangle includes a local term $h$ which acts trivially on (at least) one of its edges $e$ (this edge is marked as a double edge). Whether a star term or a plaquette term, we can connect it to the three triangle sides with three paths (dotted curves) such that $e$ belongs to one region, and the other edges belong to the two other regions.

\end{minipage}
\\
This construction gives rise to a graph $G$ 
which highly resembles $\mathcal{T}^{*}$ the dual of $\mathcal{T}$.
The vertices of $G$ consist of the chosen triangles 
center of $\mathcal{T}$, as well as the centers of 1-cells of triangles 
in $\mathcal{T}$ which are on the boundary of $\s$. Between any two vertices 
of $G$ corresponding
to the centers of two triangle $T_{1},T_{2}$ which share a side 
(i.e $T_{1}\cap T_{2}$ is a 1-cell of $\mathcal{T}$) let there be an edge; 
in addition, for every triangle which
has a side on the boundary of the surface, 
let there be an edge between the triangle
center and the boundary. The edges of $G$ are drawn on $\s$ as 
the paths constructed in the previous paragraph.

Consequently, vertices of $\mathcal{T}$
are in one-to-one correspondence with faces of $G$. Those faces induce
a partition $\mathcal{P}$ of the set of qubits $\mathcal{Q}$ according to
the face of $G$ which they belong to (if an edge of $\k$ touches
more then one face of $\mathcal{T}^{*}$ then join it to one of those
faces arbitrarily) \cite{Graph1}. We accordingly have: $\h=\bigotimes_{q\in\mathcal{Q}}\h_{q}=\bigotimes_{P\in\mathcal{P}}\h_{P}$
with $\h_{P}:=\bigotimes_{q\in P}\h_{q}$. We refer to each cluster
$P\in\mathcal{P}$ and to its Hilbert space $\h_{P}$ as a super-particle.

\begin{minipage}{0.5\textwidth}
\begin{figure}[H]	
\caption{\label{fig:superparticles} The graph embedding of $G$ in $\s$}
\includegraphics[scale=0.18]{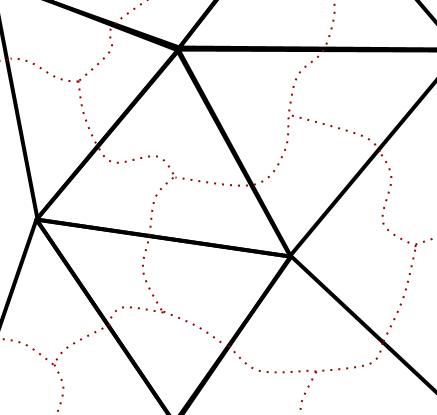}
\end{figure}
\end{minipage} \hfill
\begin{minipage}{0.45\textwidth}
\vspace{6mm}
The dark lines are the quasi-Euclidean triangulation. The dotted curves are the edges of the graph $G$ realized as the chosen paths in $\s$. The faces of $G$ induce a partition of the qubits into superparticles. The fatness of the triangles and the bounded degree of the triangulation implies that the superparticles' size is constant.
\end{minipage}

\vspace{4mm}

So far we have only used the ``slimness'' of a triangle condition in 
the definition of quasi-Euclidean condition. Here is where we need the bound on the fatness
of triangles and the upper bound on its degree. The following claims are proven in Appendix  \ref{apx:2-local}.
\begin{claim} [constant sized super-particles]
\label{claim:super-particle}
Each super-particle includes at most
$D\cdot k^{R+2}$ qubits (in particular $\mathcal{O}(1)$).
\end{claim}

\begin{claim} [punctured Hamiltonian is 2-local]
\label{claim:2-local}Each local term of $\tilde{H}$ acts on at most
two super-particles.
\end{claim}

\section{Completing the algorithm and the proofs for Theorems 
\ref{mainthm1} $\&$ \ref{mainthm2}.}
\label{sec:completingproof}
We now proof Theorems 1 \& 2. By Claims \ref{claim:super-particle}, \ref{claim:2-local}, it is possible to prepare a ground space
$\ket{\tilde{\psi}}$ of $\tilde{H}$, 
using a constant depth quantum circuit. 
Given such a groundstate, we measure every $h\in\mathcal{W}$ one by one.
Actually it will be simpler to measure $I-2\cdot\pi_{h}$ instead
where $\pi_{h}$ is the orthogonal projector onto the ground space
of $h$. Record that result of the measurement by $\lambda_{h}$. Accordingly,
having $\lambda_{h}=1$ indicates that $\ket{\tilde{\psi}}$ is already
a groundstate of $h$ whereas $\lambda_{h}=-1$ indicates an
excitation at that spot. 
The state we had $\ket{\tilde{\psi}}$
collapses by these measurements to a new state $\ket{\psi}$ which is an eigenstate of every
$h\in\mathcal{W}$, while still being in the ground space of $\tilde{H}$. Recall that the set of terms we measured (the set $\mathcal{W}$) all have access to the boundary (Definition \ref{def:access}). Thus their value can be changed via string logical operators while not effecting the 
value of any other term. This is summarized by the
following algorithm:

\subsubsection{Algorithm (constructing a groundstate for an arbitrary quasi-Euclidean $2D-CLH^*(k,2)$ instance): } \label{alg:final}
\begin{enumerate}
\item If the instance has no boundary or coboundary qubits, then it is equivalent to the toric code, so apply algorithm \ref{alg:toric} and terminate. 
\item Else, generate a groundstate of $\tilde{H}$ with a constant
depth quantum circuit.
\item For each term $h\in\mathcal{W}$ which was removed, measure $I-2\cdot\pi_{h}$, and record the measurement value as $\lambda_{h}=\pm1$. ($\pi_h$ is the orthogonal projector onto the groundspace of $h$). 
\item Fix every $h\in\mathcal{W}$ for which $\lambda_{h}=-1$: if $h$
is a star term $s$, find a path $\gamma$ from $s$ to the coboundary
and apply $L=\bigotimes_{q\in\gamma}X_{q}$. If $h$ is a plaquette
term $p$, find a copath $\gamma^*$ from $p$ to the boundary and
apply $L^{*}=\bigotimes_{q\in\gamma^{*}}Z_{q}$.
\end{enumerate}

This proves Theorem 2. Theorem 1 follows as well: if the instance has no boundary/coboundary qubits (and this can of course be checked efficiently by the verifier) then the system is equivalent to the toric code, so it's ground energy can be computed easily (see Appendix \ref{apx:Defected-toric-code} for the case where the local terms are as in the toric code only upto a factor of a scalar). Otherwise, the problem of computing the ground energy of $H$ reduces to computing the ground energy of $\tilde{H}$, since the verifier knows that any groundstate of $\tilde{H}$ can be corrected to a (possibly other) groundstate of $\tilde{H}$ such that all terms in $\mathcal{W}$ are satisfied (i.e the energy with respect to the terms in $\mathcal{W}$ is minimal). It is thus left to note that $\tilde{H}$ is a 2-local CLH, and this problem is in NP by \cite{BV}.

\section{Discussion} 
\label{sec:discussion}
An interesting property of the algorithm is that all of the {\it quantum} operations are summed up to have only {\it constant} depth. Indeed, the algorithm consists of three steps: a constant depth quantum circuit that generates a groundstate for the punctured Hamiltonian, a non-constant depth computation of path finding which can be carried out in a {\it classical} manner, and finally a constant depth quantum circuit of logical operators (tensor product of Pauli operators).

This observation regards the complexity of the algorithm, but it is interesting also conceptually. While the quantum circuit presented here is
of polynomial depth, it is enough for the verifier to obtain only a constant
depth circuit description, and verify that it is indeed a groundstate of 
the punctured Hamiltonian, 
in order to be know the ground energy of the whole system (since the verifier 
knows that these holes can always be fixed). 
This means that while the time it takes to generate a groundstate for the system is concentrated
on creating global entanglement, all the hardness and potential frustration
of the groundstate comes into play only at the level of local entanglement
of the groundstate of the 2-local punctured system.

Moreover, our results shed new light on the possible threshold phenomenon suggested in \cite{3local}. Recall that this threshold (described above in subsection \ref{subsec:prev})  regards the fact that up until $k=3,d=3$, and also for $k=2$ and arbitrary $d$, $CLH(k,d)$ always have trivial groundstate, which in turn implies that those problems are in NP. The threshold refers to the fact one cannot expect the exact same phenomenon for higher parameters since then there are systems with topological quantum order which are known to have no trivial groundstates. It is thus interesting that our proof extends this trivial state phenomenon even beyond this transition point into the regime of potentially global entanglements, in the sense that even here the prover hands us a description of a trivial state - a ground state of $\tilde{H}$ (even though it cannot in general be a groundstate of the actual instance). This raises the question of whether such a property holds for more general CLHs.

Can these results be extended to all 2D systems?   
A generalization from qubits to qudits of dimension 
larger than $2$ would imply this, under the quasi-Euclidity assumption (see Appendix \ref{apx:2d}). Thus, the main open problem is to generalize our results to higher dimensional particles. We note that in any case one can still tear holes in a regular manner (using e.g the quasi-Euclidity assumption) to obtain a punctured Hamiltonian which is 2-local with respect to superparticles, and thus has a trivial groundstate. The problem is that we do not know how to fix those holes later on: our characterization of $\sp$ instances (i.e Theorem \ref{thm:equiv}) and of fixable terms (namely, the creation of logical operators in Lemma \ref{lem:either-path-or-copath}) strongly uses the fact that the particles are 2 dimensional. It is open whether further generalization could be derived using more general characterizations 
of commuting local Hamiltonians, perhaps over general finite groups 
(e.g the quantum double model \cite{ToricCode}).

We mention that if indeed the results can be generalized to 
qudits, it might also be possible to generalize to 
3D manifolds or more, perhaps in an inductive manner.

A more technical question is whether the quasi-Euclidity condition can 
be relaxed.  
Quasi-Euclidity seems closely related to the notion of $1$-localizablity 
introduced in Hastings' paper 
\cite{Localizable} already mentioned  
(In fact, the quasi-Euclidity condition we use can be replaced by the 
technical assumption used in \cite{Localizable} regarding the girth of the 
complex; we could then deduce the existence of a groundstate 
for $\tilde{H}$ from $1$-localizabilty instead of $2$-locality). 
This raises the question of whether manifolds which are 
very non-Euclidean and which have low girths, 
can exhibit much more complex multi-particle entanglement (we mention in 
this context \cite{guth}). 

{~}

{\noindent \bf Acknowledgments}

{\noindent D.} Aharonov, O. Kenneth and I. Vigdorovich acknowledge the generous support of ERC grant number 280157 for funding their work on this paper.

{~}

\newpage

\begin{appendices}
\section{Comparison between different 2D-CLH notions}
\label{apx:2d}

\begin{defn} [$2D-CLH(k,d)$ instance]
\label{def:2D-CLH}   Consider instance $x$ of $CLH(k,d)$ for which:
\begin{enumerate}
\item There exists a two dimensional polygonal complex $\mathcal{K}$.
\item There exists a 1-1 mapping between qudits of $x$ and vertices of
$\mathcal{K}$.
\item There exists a 1-1 mapping between local terms of $x$ and faces of
$\mathcal{K}$. 
\item If $h$ corresponds to a face $f$ then the set of qudits $\{q_{1},...,q_{r}\}$
which $h$ acts on corresponds to the set of vertices $\{v_{1},...,v_{r}\}$
of $f$.
\end{enumerate}
The class of such instances is denoted by $2D-CLH(k,d)$.

\end{defn}
  Any $2D-CLH^*(k,d)$ instance can be transformed into a $2D-CLH(k,d)$ instance. Indeed, for any given $2D-CLH^*(k,d)$ instance with complex $\k$, one may define a new complex $\k^{\prime}$ by placing a vertex in the interior of every edge, and connecting two of those vertices if and only if the edges that they live on belong to the same face, and share a vertex. Then, whenever $e_1,...,e_r$ are the edges of some star/plaquette in $\k$, they are now the {\it vertices} of a face in $\k'$. Therefore we obtain a $2D-CLH(k,d)$ instance.

The reverse construction is not always
possible. To see this observe that the construction above always results with a complex where each vertex is either of degree 2 or 4 (since in the original complex each edge either  belongs to 1 plaquette or to 2 plaquettes).
However when setting $d=2$ it can be shown, using C{*}-algebraic considerations (which can found in
section 4 of \cite{3local}), that if a qubit (vertex) of a $2D-CLH(k,2)$ instance is of degree 5 or above that it must be a classical qubit. As a result,  $2D-CLH(k,2)$ instances can always be reduced to ones where the maximal degree of the complex is $4$. Consequently, when demanding that the vertices of the underlying complex of the $2D-CLH(k,2)$ instances is never of degree $3$, then the two classes happen to be equivalent (e.g in the case of a grid the two settings are precisely the same).
In other words, it is only the degree $3$ 
vertices in $2D-CLH(k,2)$ instances which prevent
them from being represented in a star-plaquette fashion.

It thus follows that our results regarding $2D-CLH^*(k,2)$ are equally true for $2D-CLH(k,2)$ as long as there are no vertices of degree $3$. If though our results are generalized to $2D-CLH^*(k,d)$ for arbitrary $d$ then they are equally true for $2D-CLH(k,d)$, with a slight modification of the quasi-Euclidity assumption. Indeed, if we assume that the surface admits a quasi-Euclidian grid (i.e a grid where each square is neither to fat nor too slim), then we may use such a grid to group qubits which are in the same square into superparticles. We thus have now a new $2D-CLH(k,d)$  defined on the grid with higher $d$. It is then left to note that a new tilted grid, with edges on the original grid's vertices, defines a $2D-CLH^*(k,d)$ instance with the exact same interaction of the original one.

\section{Algebraic definitions and lemmas \label{apx:background}}

\subsection{C{*}-algebras}

We now introduce some facts about finite dimensional C{*}-algebras
and their applications to the study of commuting local Hamiltonians. If $\mathcal{H}$ is a finite dimensional complex Hilbert space we denote
by $\mathcal{L}\left(\mathcal{H}\right)$ the complex algebra of all
linear operators on $\mathcal{H}$. 
We are in fact interested in representations of finite dimensional C*-algebras, i.e when the elements of the C*-algebras are realized as matrices. Actually, we are interested only in representations that send the unital element of the algebra to the identity operator (note that every finite dimensional C*-algebra is unital \cite{Cstar} i.e it has 
an element $\identity$ neutral to multiplication from right and left, however this element may in general be other then the identity matrix). 
Therefore, for the sake of this paper, it will be convenient to define a C*-algebra in a more reduced way then defined abstractly:
\begin{defn} [C*-algebra] Let $\h$ be a finite dimensional Hilbert space. A C*-algebra is any algebra $\a\subseteq \l (\h)$ which is closed under the $\dagger$ operation (i.e $a^\dagger \in \a$ whenever $a\in \a$) and which includes the identity operator ($I\in\a$).
	\end{defn}
Given a finite set
of operators $a_{1},...,a_{m}\in\text{\ensuremath{\mathcal{L}\left(\mathcal{H}\right)}}$,
the C{*}-subalgebra generated by $a_{1},...,a_{m}$ is the minimal
C{*}-algebra that includes $a_{1},...,a_{m}$ and is denoted by $\left\langle a_{1},...,a_{m}\right\rangle $.
This should not be confused with $Sp\{a_{1},...,a_{m}\}$ which is
the linear subspace spanned by $a_{1},...,a_{m}$ (viewed as vectors in a vector space). The dimension of
a C{*}-algebra is by definition its dimension as a vector space. The
center of an algebra $\mathcal{A}$ is by definition the subalgebra
$\mathcal{Z}(\mathcal{A})=\left\{ A\in\mathcal{A}\mid\forall B\in\mathcal{A}\,:\,\left[A,B\right]=0\right\} $.
Finally, two algebras $\a_{1},\a_{2}\subseteq\mathcal{L}(\h)$ are
said to commute if $[a_{1},a_{2}]=0$ whenever $a_{1}\in\mathcal{A}_{1},\,a_{2}\in\a_{2}$.

The structure theorem (see \cite{Cstar}) 
states that every finite dimensional C{*}-algebra
is a direct sum of algebras of all operators on a Hilbert space. One
way of formulating this is as follows: 
\begin{fact}[classification of finite dimensional C*-algebra]
\label{fact:Cstar}  Let $\mathcal{A}\subseteq\mathcal{L}\left(\mathcal{H}\right)$
a C{*}-algebra where $\mathcal{H}$ is finite dimensional. There exists
a direct sum decomposition:
\begin{equation}
\mathcal{H}=\bigoplus_{\alpha}\mathcal{H}^{\alpha}\label{eq:Cstar-directsum}
\end{equation}
 and a tensor product structure
\begin{equation} 
\mathcal{H}=\bigoplus_{\alpha}\mathcal{H}_{1}^{\alpha}\otimes\mathcal{H}_{2}^{\alpha}\label{eq:Cstar-tensor}
\end{equation}
 such that
\begin{equation}
\mathcal{A}=\bigoplus_{\alpha}\mathcal{L}\left(\mathcal{H}_{1}^{\alpha}\right)\otimes\mathcal{I}\left(\mathcal{H}_{2}^{\alpha}\right)\label{eq:Cstar-alg-dec}
\end{equation}
Furthermore, $\mathcal{Z}(\mathcal{A})$ is spanned by the set of
orthogonal projections on the $\mathcal{H}^{\alpha}$, over all the $\alpha'$s. 
\end{fact}

Here and later, $\alpha$ denotes this direct sum index (on some finite
range). $\mathcal{I}\left(\mathcal{H}_{2}^{\alpha}\right)$ here denotes
the trivial 1-dimensional algebra $\mathbb{C}\cdot I$ on $\mathcal{H}_{2}^{\alpha}$.
A proof for this fact can be found in \cite{Cstar}.

\subsection{The induced algebra by a Hamiltonian}
\begin{defn*}[induced algebra - Definition \ref{def:induced-algebra}]
\label{apxdef:induced-algebra}Let $h$ be an operator on a tensor product
Hilbert space $\mathcal{H}_{q_{1}}\otimes\mathcal{H}_{q_{2}}$ and
let $h=\sum_{i=1}^{m}h_{q_{1}}^{i}\otimes h_{q_{2}}^{i}$ be a Schmidt
decomposition\footnote{that is to say: $h_{q_{1}}^{i}\in\mathcal{L}\left(\mathcal{H}_{q_{1}}\right)$,
$h_{q_{2}}^{i}\in\mathcal{L}\left(\mathcal{H}_{q_{2}}\right)$  for each $i$ and the that sets $\left\{h^i_{q_1}\right\}^m_{i=1}$, $\left\{h^i_{q_2}\right\}^m_{i=1}$ are orthogonal with respect to the Hilbert-Schmidt inner product i.e $tr({h^i_{q_l}}^\dagger\cdot{h^j_{q_l}})=0$ for any $i\ne j$ and $l=1,2$)} of $h$.
The induced algebra of $h$ on $\mathcal{H}_{q_{1}}$ is denote by
$\mathcal{A}_{\mathcal{H}_{q_{1}}}^{h}$ or in short $\mathcal{A}_{q_{1}}^{h}$
and is defined to be the C*-algebra  $\langle h^i_{q_1}\rangle^m_{i=1}$ (note that by definition $I\in \a^{h}_{q_1}$).
\end{defn*}

\begin{claim}[induced algebra is independent on decomposition]
\label{inducedalgebra_indepednent}
In the case where $h$ is Hermitian, the induced algebra is independent
on the chosen decomposition. In fact, even if it is not a Schmidt
decomposition, the algebra $\a_{q_{1}}^{h}$ will remain the same
as long as the set $\left\{ h_{2}^{i}\right\} _{i}$ is linearly independent.
\end{claim}

\begin{proof}
Write two decompositions:
\[
\sum_{i}h_{q_{1}}^{i}\otimes h_{q_{2}}^{i}=h=\sum_{j}\hat{h}_{q_{1}}^{j}\otimes\hat{h}_{q_{2}}^{j}
\]
where the sets $\left\{ h_{q_{2}}^{i}\right\} _{i},\left\{ \hat{h}_{q_{2}}^{j}\right\} _{j}$
are linearly independent. We show that the induced algebra according
to the first decomposition $\a$ is contained in the induced algebra
according to the second decomposition $\hat{\a}$. Then, the equality
follows by symmetry. Let us complete the set $\left\{ h_{q_{2}}^{i}\right\} _{i}$
into a basis of $\mathcal{L}\left(\h_{q_{2}}\right)$, by the operators
$\left\{ h_{q_{2}}^{i'}\right\} _{i'}$. We can thus write the $\hat{h}_{q_{2}}^{j}$
operators in terms of this basis: 
\[
\hat{h}_{q_{2}}^{j}=\sum_{i}c_{i,j}h_{q_{2}}^{i}+\sum_{i^{\prime}}c_{i',j}h_{q_{2}}^{i^{\prime}}
\]
 with $c_{i,j}$, $c_{i',j}$ complex numbers. Setting the equations
above, in the second decomposition we get:
\[
h=\sum_{j}\hat{h}_{q_{1}}^{j}\otimes\hat{h}_{q_{2}}^{j}=\sum_{i}\left(\sum_{j}c_{i,j}\hat{h}_{q_{1}}^{j}\right)\otimes h_{q_{2}}^{i}+\sum_{i'}\left(\sum_{j}c_{i',j}\hat{h}_{q_{1}}^{j}\right)\otimes h_{q_{2}}^{i'}
\]
 Comparing this with $h=\sum_{i}h_{q_{1}}^{i}\otimes h_{q_{2}}^{i}$
we conclude by linear independence that $h_{q_{1}}^{i}=\sum_{j}c_{i,j}\hat{h}_{q_{1}}^{j}$
for any $i$ which mean that each $h_{q_{1}}^{i}$ is in $\hat{\a}$
and so $\a\subseteq\hat{\a}$ 

\end{proof}
Note that also if a Hamiltonian $h$ acts on multiple particles $\mathcal{H}=\bigotimes_{i=1}^{n}\mathcal{H}_{q_{i}}$
then we can combine a subset of particles together, that is write
$\mathcal{H}=\left(\bigotimes_{i=1}^{r}\mathcal{H}_{q_{i}}\right)\otimes\left(\bigotimes_{i=r+1}^{n}\mathcal{H}_{q_{i}}\right)$
and speak of the algebra that $h$ induces on the first factor. By
the notation given in definition \ref{def:induced-algebra}, this
is simply the algebra $\mathcal{A}_{\bigotimes_{i=1}^{r}\mathcal{H}_{q_{i}}}^{h}$
which we will denote in short by $\mathcal{A}_{q_{1},...,q_{r}}^{h}$.
We call this the algebra that $h$ induces on the particles $q_{1},...,q_{r}$.
\begin{lem}[connection between the induced algebra of a system and its subsystems]
\label{lemma:subalgebra}Suppose $h$ is a Hamiltonian acting on
$\mathcal{H}=\mathcal{H}_{q_{1}}\otimes\mathcal{H}_{q_{2}}\otimes\mathcal{H}_{q_{3}}$.
Then $\mathcal{A}_{q_{1},q_{2}}^{h}\subseteq\mathcal{A}_{q_{1}}^{h}\otimes\mathcal{A}_{q_{2}}^{h}.$ 
\end{lem}

\begin{proof}
Write:
\begin{equation}
h=\sum_{i=1}^{m}h_{q_{1},q_{2}}^{i}\otimes h_{q_{3}}^{i}\label{eq:subalgebra-lemma-schmidt1}
\end{equation}
 a Schmidt decomposition of $h$ where $h_{q_{1},q_{2}}^{i}$ act
only on $\mathcal{H}_{q_{1}}\otimes\mathcal{H}_{q_{2}}$ and $h_{q_{3}}^{i}$
act only on $\mathcal{H}_{q_{3}}$. So $\mathcal{A}_{q_{1},q_{2}}^{h}$
is generated by $\left\{ h_{q_{1},q_{2}}^{i}\right\} _{i=1}^{m}$.
Now, for every $1\leq i\leq m$ write a Schmidt decomposition for
$h_{q_{1},q_{2}}^{i}$ in respect to the tensor product $\mathcal{H}_{q_{1}}\otimes\mathcal{H}_{q_{2}}$:
\begin{equation}
\,\,\,h_{q_{1},q_{2}}^{i}=\sum_{j=1}^{r_{i}}h_{q_{1}}^{i,j}\otimes h_{q_{2}}^{i,j}\label{eq:subalgebra-lemma-schmidt2}
\end{equation}
So we now have that:
\begin{equation}
h=\sum_{i=1}^{m}h_{q_{1},q_{2}}^{i}\otimes h_{q_{3}}^{i}=\sum_{i=1}^{m}\sum_{j=1}^{r_{i}}h_{q_{1}}^{i,j}\otimes h_{q_{2}}^{i,j}\otimes h_{q_{3}}^{i}\label{eq:subalgebra-lemma-combine}
\end{equation}
The set $\{h_{q_{3}}^{i}\}_{i=1}^{m}$ is linearly independent, and
for each $1\leq i\leq m$ the set $\left\{ h_{q_{2}}^{i,j}\right\} _{j=1}^{r_{i}}$
is linearly independent. It follows that the set $\left\{ h_{q_{2}}^{i,j}\otimes h_{q_{3}}^{i}\right\} _{i,j=1}^{m,r_{i}}$
is linearly independent, and therefore by Claim 
\ref{inducedalgebra_indepednent}, $\left\{ h_{q_{1}}^{i,j}\right\} _{i,j=1}^{m,r_{i}}$ generates $\mathcal{A}_{q_{1}}^{h}$ and in particular $\left\{ h_{q_{1}}^{i,j}\right\} _{i,j=1}^{m,r_{i}}\subseteq \mathcal{A}_{q_{1}}^{h}$.
In a similar way we also get that $\left\{ h_{q_{2}}^{i,j}\right\} _{i,j=1}^{m,r_{i}}\subseteq \mathcal{A}_{q_{2}}^{h}$.
It is only left to note that
$h_{q_{1},q_{2}}^{i}$ belongs to $\mathcal{A}_{q_{1}}^{h}\otimes\mathcal{A}_{q_{2}}^{h}$
as can readily be seen in eq. \ref{eq:subalgebra-lemma-schmidt2}
and conclude that ${\mathcal{A}^{h}_{q_{1},q_{2}}}\subseteq{\mathcal{A}^{h}_{q_{1}}}\otimes{\mathcal{A}^{h}_{q_{2}}}$
as $\left\{ h_{q_{1},q_{2}}^{i}\right\}$ generates $\mathcal{A}_{q_{1},q_{2}}^{h}$.

\end{proof}

\begin{lem} [commutation of induced algebras]
\label{lemma:alg-commute}Let $l$ and $r$ be two (not necessarily
commuting) Hamiltonians acting on the Hilbert space $\mathcal{H}=\mathcal{H}_{q_{1}}\otimes\mathcal{H}_{q_{2}}\otimes\mathcal{H}_{q_{3}}$
such that $l$ acts only on $\mathcal{H}_{q_{1}}\otimes\mathcal{H}_{q_{2}}$
(and trivially on $\mathcal{H}_{q_{3}})$ and $r$ acts only on $\mathcal{H}_{q_{2}}\otimes\mathcal{H}_{q_{3}}$
(and trivially on $\mathcal{H}_{q_{1}})$. Then l and $r$ commute
if and only if $\mathcal{A}_{q_{2}}^{l}$ and $\mathcal{A}_{q_{2}}^{r}$
commute.
\end{lem}

\begin{proof}

Write a Schmidt-decomposition for $l$ and $r$:
\begin{equation}
l=\sum_{i}l_{q_{1}}^{i}\otimes l_{q_{2}}^{i},\,\,\,\,r=\sum_{j}r_{q_{2}}^{j}\otimes r_{q_{3}}^{j}\label{eq:l-and-r-schmidt}
\end{equation}
 So:
\begin{equation}
\left[l,r\right]=\sum_{i,j}l_{q_{1}}^{i}\otimes\left[l_{q_{2}}^{i},r_{q_{2}}^{j}\right]\otimes r_{q_{3}}^{j}
\label{eq:l-and-r-commute}
\end{equation}
It is clear then that if $\mathcal{A}_{q_{2}}^{l}$ and $\mathcal{A}_{q_{2}}^{r}$
commute then $[l,r]=0$ since $l_{q_{2}}^{i}\in\mathcal{A}_{q_{2}}^{l}$
and $r_{q_{2}}^{j}\in\mathcal{A}_{q_{2}}^{r}$. For the converse,
note that $\left\{ l_{q_{1}}^{i}\right\} _{i}$ is a linearly independent
set and so is $\left\{ r_{q_{3}}^{j}\right\} _{j}.$ Hence $\left\{ l_{q_{1}}^{i}\otimes r_{q_{3}}^{j}\right\} _{i,j}$
is a linearly independent set. Consequently, $[l,r]=0$ can only occur
if $\left[l_{q_{2}}^{i},r_{q_{2}}^{j}\right]=0$ for all $i,j$ which
in turn implies that $\mathcal{A}_{q_{2}}^{l}$ and $\mathcal{A}_{q_{2}}^{r}$
commute.

\end{proof}

\begin{lem} [full operator algebra implies triviality]
\label{lem:schurs-lemma}In the case where $l$ and $r$ commute,
if the algebra $\mathcal{A}_{q_{2}}^{l}$ is the full operator algebra
$\mathcal{L}(\mathcal{H}_{q_{2}})$ then the algebra $\mathcal{A}_{q_{2}}^{r}$
is the trivial algebra $\mathcal{I}(\mathcal{H}_{q_{2}})=\mathbb{C}\cdot I$.
\end{lem}

\begin{proof}
The algebra $\mathcal{A}_{q_{2}}^{l}$ commutes with $\mathcal{A}_{q_{2}}^{r}$,
but since the center of a full operator algebra is trivial we conclude
that $\mathcal{A}_{q_{2}}^{r}=\mathbb{C}\cdot I$.
\end{proof}

\subsection{The induced algebra on a qubit}

Fact \ref{fact:Cstar} makes it easy to characterize C{*}-algebras
on a single qubit - that is subalgebras of $\mathcal{L}(\mathbb{C}^{2})$.
\begin{lem} [induced algebras on a qubit]
\label{lemma:qubit-algebra}A C{*}-algebra $\mathcal{A}$ on a qubit
is either 1 dimensional (the trivial algebra), 4 dimensional (the
full operator algebra $\mathcal{\mathcal{A}=L}(\mathbb{C}^{2})$),
or a 2 dimensional algebra - in which case $\mathcal{A}$ is generated
by a single operator $P$ which is a traceless Pauli operator. Furthermore,
if $\a'$ is another two dimensional C{*}-algebra which commutes with
$\mathcal{A}$ then $\mathcal{A}=\a'$.
\end{lem}

\begin{proof}
By fact \ref{fact:Cstar}, $\mathcal{A}$ admits a direct sum decomposition
of full operator algebras. Since $\dim\mathcal{H}=2$ then the dimensions
of the decomposition components as in eq. \ref{eq:Cstar-directsum},\ref{eq:Cstar-tensor}
can take exactly three forms: if there is no direct sum so we can
write $\mathcal{H}=\mathcal{H}_{1}\otimes\mathcal{H}_{2}$ and then
either $\dim\mathcal{H}_{1}=1\,\text{ and }\dim\mathcal{H}_{2}=2$
or $\dim\mathcal{H}_{1}=2\text{ and }\dim\mathcal{H}_{2}=1$. The
first implies that $\dim\mathcal{A}=1$ and the second implies that
$\dim\mathcal{A}=4$ (eq. \ref{eq:Cstar-alg-dec}). The third and
last option is that the direct sum is of length 2 and that each $\mathcal{H}_{1}^{\alpha}$
is of dimension 1, and this implies that $\dim\mathcal{A}=2$. 

In the latter case, there must be some $P\in\mathcal{A}\backslash\mathbb{C}\cdot I$.
We may assume that $P$ is Hermitian for the following reasoning:
if $P$ is anti-Hermitian ($P^\dagger=-P$) 
then $iP$ is an Hermitian operator in $\mathcal{A}\backslash\mathbb{C}\cdot I$, 
and if $P$ is not anti-Hermitian then $P+P^{\dagger}$ is a non-zero 
Hermitian operator in $\mathcal{A}\backslash\mathbb{C}\cdot I$. We may also
assume that $P$ is traceless for if it isn't then $P-\frac{1}{2}tr(P)\in\mathcal{A}\backslash\mathbb{C}\cdot I$
is. By normalizing $P$ we may assume that $P$ is also unitary.

Finally, if $\a'$ commutes with $\a$ and is generated by some traceless 
Pauli operator $P'$, then $[P,P']=0$ implying that $P=\pm P'$ and
so $\a=\langle P\rangle=\langle\pm P\rangle=\a'$

\end{proof}

\begin{remark}
We distinguished between the term Pauli operator and Pauli matrix.
A Pauli matrix refers to any of the four matrices $I,X,Y,Z$, whereas
a Pauli operator is a linear operator which is represented by a Pauli
matrix in some orthonormal basis. Note that a traceless operator
on a qubit is Pauli if and only if it is Hermitian and unitary. 
Indeed any traceless Hermitian and unitary operator on a qubit has two 
eigenvalues 1 and -1 and so it is represented by a Pauli matrix for some 
orthonormal basis.
\end{remark}

\subsection{A few more simple but useful lemmas}

\begin{lem}[anti-commutation lemma]
\label{lem:anti-commute}
Let $C,D$ be two anti-commuting non-zero opertors on a single qubit. 
Then one of the following holds:
\begin{enumerate}
\item $C$ and $D$ are both invertable, in which case $C=cZ$ and $D=dX$ $(c,d\ne 0)$ for some choice of basis. In this case we say that $C$,$D$ anti-commute \textbf{regularly}.
\item $C$ and $D$ are proportional to complementary orthogonal projectors: that is $C=c(I+Z)$ and $D=d(I-Z)$ for some choice of basis $(c,d\ne 0)$. In this case we say that $C$,$D$ anti-commute \textbf{irregularly}.
\end{enumerate}
\end{lem}

\begin{proof}

First choose a diagonalizing basis for $C$ so we have $C\in \langle Z \rangle$.
\textbf{Case 1:}
If $D$ is invertible we may write $DCD^{-1}=-C$. Hence $C$ is traceless and since $C\ne 0$ it follows that $C$ is invertible as well. Hence $CDC^{-1}=-D$ and so $D$ is also traceless. Thus up to multiplication by non-zero scalars (namely $c,d$) we have that $C$ and $D$ have eigenvalue values 1,-1. Therefore, up to a choice of basis $C=Z$. To show that we may choose the basis such that not only $C=Z$ but also $D=X$ we need to find a unitary $U$ such that $U^{\dagger}CU=U^{\dagger}ZU=Z$ and such that $U^{\dagger}DU=X$. Write $D$ in the Pauli
basis: $D=d_{x}X+d_{y}Y$ with $d_{x}^{2}+d_{y}^{2}=1$ (it does
not have an $I$ component since it is traceless and it does not have a $Z$ component since $\left\{ Z,D\right\} =0$) and choose
$U=\left|0\right\rangle \left\langle 0\right|+\mu\left|1\right\rangle \left\langle 1\right|$
where $\mu=d_{x}+id_{y}$ . Clearly $U$ commutes
with $Z$ so it is left show that $U^{\dagger}DU=X$, indeed:
\begin{equation}
U^{\dagger}DU\left|0\right\rangle =U^{\dagger}\left(d_{x}X+d_{y}Y\right)\left|0\right\rangle = U^{\dagger}\left(d_{x}\left|1\right\rangle +id_{y}\left|1\right\rangle \right)= \mu U^\dagger \ket 1 = \ket 1
\end{equation}
and:
\begin{equation}
U^{\dagger}DU\left|1\right\rangle =\mu U^{\dagger}\left(d_{x}X+d_{y}Y\right)\left|1\right\rangle =\mu U^{\dagger}\left(d_{x}\left|0\right\rangle -id_{y}\left|0\right\rangle \right)= \mu \cdot \mu^{-1} U^\dagger \ket 0 = \ket 0
\end{equation}
Therefore $U^{\dagger}DU=X$.
\textbf{Case 2:}
Otherwise, $D$ is not invertible, and therefore $C$ is not invertible as well (by the above). Since they are both non-zero it follows that have rank 1, i.e they are orthogonal projectors on a 1-dimensional subspace, and in particular $CD$ is of rank at most 1. The fact that $CD=-DC$ implies that $CD$ is traceless and so it follows that $CD$ cannot be of rank 1 meaning that $CD=0$, in particular $C$,$D$ commute. It follows that in their diagonalizing basis $C$ and $D$ are complementary orthogonal projectors in the computational basis (up to multiplication by a non-zero scalar) and so we may choose a basis for which $C=I+Z$, $D=I-Z$.

\end{proof}

\begin{lem}[unitary equivalence of Pauli group representations]
\label{lem:unitaryequiv}
Let $C$,$D$ be (any) two operators on a single qubit. Then for some choice of basis $C\in Sp\{I,Z\}$ and $D\in Sp\{I,X,Z\}$
\end{lem}

\begin{proof}

In the diagonalizing basis of $C$ we have that $C=c_1I+c_2Z$ and $D=d_1I+d_2Z+d_3\tilde{D}$ where $\tilde{D}\in Sp\{X,Y\}$. If $\tilde{D}=0$ we are done. Else, 
we may apply \ref{lem:anti-commute} on $Z$, $\tilde{D}$.
\end{proof}

\section{Classical qubits}
\label{apx:classicalQubits}
\begin{proof} [Proof of claim \ref{claim:no-classical-qubits}]
As long as non-trivial classical qubits exist, the prover (namely, Merlin) chooses such a qubit $q$, computes the decomposition $\mathcal{H}_{q}=\bigoplus_{\alpha_{q}}\mathcal{H}_{q}^{\alpha_{q}}$,
chooses a value of $\alpha_{q}$ such that $\mathcal{H}^{\alpha_{q}}=\left(\bigotimes_{q'\ne q}\mathcal{H}_{q'}\right)\otimes\h_{q}^{\alpha_{q}}$
includes a groundstate for $H$, and saves the corresponding orthogonal
projector $\pi^{\alpha_{q}}:\mathcal{H}_{q}\rightarrow\h_{q}^{\alpha_{q}}$.
Merlin then restricts all local terms acting on $q$ to that subspace $h\rightarrow h^{\alpha_{q}}:=h\pi^{\alpha_{q}}=\pi^{\alpha_{q}}h$
and obtains a new $\sp$ instance.
 Since on every step the Hilbert
space dimension decreases, this process must terminate. Merlin then
sends the verifier (namely, Arthur) the sequence of projectors collected
on this process, a total of $\mathcal{O}(n)$ bits. 
 One by one, Arthur computes
the restriction of each local term according to one of the projectors
$h\rightarrow h^{\alpha_{q}}$ just as Merlin did (actually order
doesn't matter since all the projectors commute). At this point Arthur
and Merlin have in hand a new instance $x'$ of $\sp$ with corresponding Hamiltonian $H^\prime$ where the only
classical qubits are the trivial ones. 
 \\
\textbf{Soundness: }The local Hamiltonian $H'$ of $x'$ is simply
a restriction of $H$ to the subspace of $\mathcal{H}$ given by the
product of all projectors. That is, there exists $\tilde{\alpha}=\left(\alpha_{q}\right)_{q}$
where $q$ ranges over all the classical qubits chosen in this process,
such that $H'=H^{\tilde{\alpha}}$. It follows that any eigenstate of $H^{\tilde{\alpha}}$, namely $\ket{\psi} \in \h ^{\tilde{\alpha}}$ such that $H^{\tilde{\alpha}}\ket{\psi}=\lambda\ket{\psi}$,
is also an eigenstate of $H$ with the same eigenvalue: $H\ket{\psi}=\sum_{\alpha}H^{\alpha}\ket{\psi}=H^{\tilde{\alpha}}\ket{\psi}=\lambda\ket{\psi}$.
\\
\textbf{Completeness:} Suppose $H\ket{\psi}=\lambda\ket{\psi}$ then
we may write $\ket{\psi}=\sum_{\alpha}\ket{\psi^{\alpha}}$ where $\ket{\psi^\alpha}\in \h ^\alpha$ and thus
$\sum_{\alpha}H^{\alpha}\ket{\psi^{\alpha}}=H\ket{\psi}=\lambda\ket{\psi}=\sum_{\alpha}\lambda\ket{\psi^{\alpha}}$.
This implies that $H^{\alpha}\ket{\psi^{\alpha}}=\lambda\ket{\psi^{\alpha}}$
for every $\alpha$. Hence in the process defined above, Merlin can indeed choose at every step an $\alpha$ for which
$\ket{\psi^{\alpha}}\ne0$ and obtain a groundstate for the restricted
Hamiltonian: $H\ket{\psi^{\alpha}}=H^{\alpha}\ket{\psi^{\alpha}}=\lambda\text{\ensuremath{\ket{\psi^{\alpha}}}}$.
This means that for any eigenvalue $\lambda$ of $H$ there exists
a choice of a projector at every step for which $\lambda$ is also
an eigenvalue of $H^{\alpha}$.
\end{proof}

\section{Logical operators}
\label{apx:logical}
\begin{defn} [logical operators]
\label{def:logical-operators}
Let $H=\sum_i{h_i}$ a commuting local Hamiltonian. A unitary operator $L$ is called a {\it logical operator} of $H$ if $L\left( GS(H) \right) \subseteq GS(H)$. (in fact in this case the equality holds since $L$ is unitary)
\end{defn}
A simple and well known fact is:
\begin{claim} [logical operators]
\label{claim:logical-operators}
Let $H=\sum_i{h_i}$ a commuting local Hamiltonian and let $L$ be a unitary operator. If $L$ commutes with each $h_i$ then $L$ is a logical operator.
\end{claim}
\begin{proof}
Consider a basis for the Hilbert space consisting of states which are eigenstates for all $h_i$ simultaneously. Let $\ket \psi$  be a basis element which is also a groundstate of $H$. For each $i$ there exists $\lambda_i \in \mathbb{R}$ such that $h_i \ket \psi = \lambda_i \ket \psi$. It follows that: 
$$ h_i\cdot (L\ket\psi)=L\cdot (h_i \ket \psi) =L\cdot (\lambda_i \ket \psi) = \lambda_i (L \ket\psi) $$ 
and so $L\ket\psi$ is a groundstate as well.
\end{proof}
 
To device the logical operator for our case, recall the notion of {\it string operators}. These are operators which are defined on paths or copaths on the complex as tensor products of Pauli matrices over the qubits of the path. For example, consider the case of two plaquette terms $B_{p_1},B_{p_2}$ removed from the surface. Connect these two plaquettes with a copath $\gamma^*$ and consider the operator $L=\bigotimes_{q\in\gamma^*}Z_q$. Clearly $L$ commutes with all Hamiltonians, except for $B_{p_1},B_{p_2}$. In fact $L$ anti-commutes with  $B_{p_1}$ and $B_{p_2}$ and therefore changes their value. Indeed if for some $\ket \psi $ we have that $B_{p_i}\ket\psi=\lambda_i\ket\psi$ ($i=1,2$) so: $$ B_{p_i}\cdot (L\ket\psi)= - L\cdot (B_{p_i} \ket \psi) =- L\cdot (\lambda_i \ket \psi) = -\lambda_i (L \ket\psi) $$.

\section{Defected toric code}
\label{apx:Defected-toric-code}
By a defected toric code we understand a system which is a
slight variation of the usual toric code. In such systems we allow $A_{s}$ to be
any non-trivial element of the algebra $\left\langle Z^{\otimes |s|} \right\rangle $,
and allow $B_p$ to be any non-trivial element of the algebra
$\left\langle X^{\otimes |p|}\right\rangle $. This is the same as saying that
for any $s$ and any $p$ there exists $u_{s},v_{p}$ and $u_{s}^{\prime},v_{p}^{\prime}\ne0$
such that ${A}_{s}=u_{s}I^{\otimes|s|}+u_{s}^{\prime}Z^{\otimes|s|}$
and ${B}_{p}=v_{p}I^{\otimes|p|}+v_{p}^{\prime}X^{\otimes|p|}$.
Clearly, the addition of identity makes no difference, it is just
a constant addition of energy so we can completely ignore it. Also,
choosing other positive values for $u_{s}^{\prime},v_{p}^{\prime}$
which are not $1$ as it is for $A_{s},B_{p}$ is nothing but rescaling
the energy values which makes no difference as well in respect to
the system's eigenstates. 

The case where some of the $u_{s}^{\prime},v_{p}^{\prime}$
are negative instead of positive is a bit more subtle.
For any pair of stars $s_{1},s_{2}$ with $u_{s_{1}}^{\prime},u_{s_{2}}^{\prime}<0$,
we may connect them by a path $\gamma$ and apply $L=\bigotimes_{q\in\gamma} X_q$ right up start.
This operation leaves all operators in tact except 
${A}_{s_{1}},{A}_{s_{2}}$
where now the sign of $u_{s_{1}}^{\prime},u_{s_{2}}^{\prime}$ is
flipped. This can of course be done also for any pair of plaquettes
by connecting them with a copath $\gamma^{*}$ and by applying $L^*=\bigotimes_{q\in\gamma^*} Z_q$.
Consequently, after repeating this for every pair we obtain a new
system where at most one $s$ and at most one $p$ have negative values
for $u_{s}^{\prime}$,$v_{p}^{\prime}$. 

In order to know what is the ground energy of the system, one only needs to count the number
of defected star terms $N_{stars}$, as well as the number of defected plaquette
terms $N_{plaq}$, and to check whether those numbers are even or odd.
 If they are both even then the system is frustration-free and so the ground energy is simply $-\sum_s (u_s+|u^\prime_s|)-\sum_p (v_p+|v^\prime_p|)$. 
If $N_{stars}$ is odd this means that 
always one star will be unsatisfied. The ground energy coming from the star 
terms will be given by violating the star $s_{min}$ 
such that $|u^\prime_s|$ is minimal. 
It follows that the ground energy from the star terms in this case is 
$-\sum_{s} (u_s+|u^\prime_s|)+2|u^\prime_s|$. 
The same argument holds when  $N_{plaq}$ is odd.

\section{Equivalence to the toric code \label{apx:equiv}}

This section is focused on proving Theorem \ref{thm:equiv}.

\begin{defn} [path/copath]
	\label{def:path-copath}A path is a sequence of stars $(s_{0},...,s_{r})$
	such that $s_{i-1}$ and $s_{i}$ share an edge $e_{i}$ for each
	$i=1,...,r$. This path is associated with the sequence of edges $(e_{1},...,e_{r})$.
	A copath is a sequence of plaquettes $(p_{0},...,p_{r})$ such that
	$p_{i-1}$ and $p_{i}$ share an edge $e_{i}$ for each $i=1,..,r$.
	This copath is associated with the sequence of edges $(e_{1},..,e_{r})$.
\end{defn}

Note that Definition \ref{def:path-copath} in particular includes paths (copaths)
where all the edges belong to one particular
plaquette (star).

We start by analyzing the interactions locally by considering a plaquette
$p$ and an adjacent star $s$ which share the qubits $q_{1},q_{2}$.
The proof is strongly based on the study of the induced algebras on
those qubits - $\mathcal{A}_{q_{1}}^{p},\mathcal{A}_{q_{2}}^{p},\mathcal{A}_{q_{1}}^{s},\mathcal{A}_{q_{2}}^{s}$
as well as the induced algebras on the Hilbert space of both qubits
- $\mathcal{A}_{q_{1},q_{2}}^{p}$, $\mathcal{A}_{q_{1},q_{2}}^{s}$.

\subsection{Induced algebras on single qubits in the interior}

\begin{claim} [2-dim algebras in the interior]
\label{claim:interior-2-dim}Every star induces a 2-dimensional algebra
on each of its qubits except those in the coboundary. Similarly,
every plaquette induces a 2-dimensional algebra on each of its qubits
except those in the boundary.
\end{claim}

\begin{proof}
If a star term $A_s$ acts on a qubit $q$ which is not in the coboundary
then this means that besides $A_s$, $q$ is acted upon non-trivially by
some other star term $A_s'$. Since $q$ is the only qubit that
$s$ and $s'$ share, we may conclude by Lemma \ref{lem:schurs-lemma}
that both $s$ and $s'$ induce neither a trivial algebra nor the
full operator algebra on $q$. According to Lemma \ref{lemma:qubit-algebra}
this means that they both induce a 2 dimensional algebra. The proof for
plaquettes is the same.
\end{proof}
Lemma \ref{lemma:qubit-algebra} 
also tells us that the 2-dimensional
algebras above are generated by a single Pauli operator. Thus we denote
by $a_{q}^{s}$ a Pauli operator that generates $\mathcal{A}_{q}^{s}$
and by $b_{q}^{p}$ a Pauli operators that generates $\mathcal{A}_{q}^{p}$.
This has significance only when the the induced algebra on $q$
is 2-dimensional. In the following, when we use the notation $a_{q}^{s}$ (or $b_{q}^{p})$
it is implied that the relevant induced subalgebra on 
$q$ is 2-dimensional. 
We thus derive, using the above plus Lemma \ref{lemma:qubit-algebra}, 
that 
 if $q$ is not in the coboundary, then the algebras induced by the two stars 
acting on $q$ are both equal, and up to a change of basis can be written 
as $\langle Z \rangle$; and  if $q$ is not in the boundary, then the algebras induced 
by the two plaquettes acting on it are equal, and can be written as 
$\langle P \rangle$ for some Pauli operator $P$. 

We now clarify the connection between these two algebras: 
\begin{claim} [commutation entails classicality] 
\label{claim:commute-entails-classical}
If a star $s$ and a plaquette $p$ both induce on $q$ $2$-dimensional 
subalgebras, generated by 
$a_{q}^s, b_q^p$ respectively, then $[a_{q}^{s},b_{q}^{p}]\ne0$.
\end{claim}

\selectlanguage{english}%
\begin{proof}
Suppose (by contradiction) that $[a_{q}^{s},b_{q}^{p}]=0$. This implies
that $\mathcal{A}_{q}^{s}=\left\langle a_{q}^{s}\right\rangle =\left\langle b_{q}^{p}\right\rangle =\mathcal{A}_{q}^{p}$.
If we let $s',p'$ denote the other two star and plaquette acting
on $q$, so we have (by lemma \ref{lemma:qubit-algebra}) that 
$\mathcal{A}_{q}^{s'}$ is either the identity or equal to $\mathcal{A}_{q}^{s}$,  
and likewise $\mathcal{A}_{q}^{p'}$ is either equal to the identity or to 
$\mathcal{A}_{q}^{p}$ and therefore all
four induced algebras are the same and equal to $\mathcal{A}=\left\langle a_{q}^{s}\right\rangle $ (apart from maybe some of them being trivial).
Since $a_{q}^{s}$ is non degenerate, its two distinct spectral projections
are in the center of $\mathcal{A}$. Therefore $\mathcal{A}$ decomposes
to a direct sum of two 1-dimensional subspaces, which are invariant
\foreignlanguage{american}{under $A_{s},A_{s'},B_{p},B_{p'}$} meaning
that $q$ is a classical qubit. However, all classical qubits have 
been turned into trivial qubits (as shown in Subsection \ref{subsec:classical}), and so such a qubit would not have a star and plaquette inducing on 
it a $2$ dimensional algebra - contradiction.   
\end{proof}

\selectlanguage{american}%
\begin{claim} [change of basis]
\label{claim:changeofbasis1} 
With the same conditions as in Claim \ref{claim:commute-entails-classical}, 
there is a choice of basis for $\mathcal{H}_{q}$
such that $a_{q}^{s}=Z$ and $b_{q}^{p} = aZ+bX$ with $b\ne 0$. 
\end{claim}
\begin{proof}
First choose a basis for which $a_{q}^{s}=Z$. Then write $b_{q}^{p}=aZ+bQ$
where $Q\in Sp\{X,Y\}$. Since $\left\{ Q,Z\right\} =0$ we can use
Lemma \ref{lem:unitaryequiv} to conclude that for some change
of basis $Q=X$ while keeping $a_{q}^{s}=Z$. The fact that $b\ne 0$ follows from claim \ref{claim:commute-entails-classical}. 
\end{proof}
 
Thus we can assume that for every star term which induces a 2-dimensional 
algebra on some qubit, this algebra is $\langle Z \rangle$, 
whereas whenever a plaquette term induces a 2-dimensional algebra on 
a qubit then it is the algebra $\langle P \rangle$ for some Pauli
$P \in Sp\{X,Z\}$ (depending on the qubit).

\subsection{Induced algebras on two qubits}

In this subsection we prove the following lemma:
\begin{lem} [induced algebras on two qubits]
	\label{lem:strong-commutation} 
	Let $s,p$ be adjacent star and plaquette, both acting on the two qubits $q_{1},q_{2}$,	
	such that $\dim\a_{q_{1}}^{s}=\dim\a_{q_{2}}^{s}=\dim\a_{q_{1}}^{p}=2$. Then (up to a choice of basis for the qubits): $\a_{q_{1}}^{s}=\a_{q_{2}}^{s}=\left\langle Z\right\rangle $,
	$\a_{q_{1}}^{p}=\left\langle X\right\rangle $ and moreover $\a_{q_{1},q_{2}}^{s}=\left\langle Z\otimes Z\right\rangle $. 
If in addition $\dim\a_{q_{2}}^{p}=2$
then $\a_{q_{2}}^{p}=\left\langle X\right\rangle $ as well, 
and moreover $\a_{q_{1},q_{2}}^{p}=\left\langle X\otimes X\right\rangle $.
\end{lem} 

\selectlanguage{english}%
\begin{proof}
	Let $P$ be a Pauli operator such that $\a_{q_{1}}^{p}=\left\langle P\right\rangle $.
	We start by assuming that $Z\otimes Z\notin\a_{q_{1},q_{2}}^{s}$.
	A contradiction will imply that $\left\langle Z\otimes Z\right\rangle \subseteq\a_{q_{1},q_{2}}^{s}$. 
By Lemma \ref{lemma:subalgebra}:
\begin{equation}\label{eq:zz}
\mathcal{A}_{q_{1}.q_{2}}^{s}\subseteq\mathcal{A}_{q_{1}}^{s}\otimes\mathcal{A}_{q_{2}}^{s}=Sp\{I\otimes I,Z\otimes I,I\otimes Z,Z\otimes Z\}
\end{equation}
As $I\otimes I$ is by definition a member of the algebra $\mathcal{A}_{q_{1},q_{2}}^{s}$
must include some other element of the form:
\begin{equation}
T=uZ\otimes I+vI\otimes Z+wZ\otimes Z
\end{equation}
where $u,v,w\in\mathbb{C}$ are not all zero.
\selectlanguage{american}%
\begin{claim}
	\label{claim:strongcommutation-uvw-are-real}We may assume $u,v,w\in\mathbb{R}$ 
\end{claim}

\begin{proof}
	$T+T^{\dagger}$ is Hermitian and belongs to $\mathcal{A}_{q_{1},q_{2}}^{s}$
	so we are done unless it is 0. In that case, $T$ is anti-Hermitian
	(i.e $T=-T^{\dagger}$) in which case we choose $iT$ which is Hermitian
	and non-zero. Since $Z\otimes I,I\otimes Z,Z\otimes Z$ are Hermitian,
	choosing $T$ to be Hermitian implies that $u,v,w\in\mathbb{R}$. 
\end{proof}
\selectlanguage{english}%

\begin{claim}\label{claim:strongcommutation-no-trivial}
$Z\otimes I$ and $I\otimes Z$
	are not members of $\mathcal{A}_{q_{1},q_{2}}^{s}$
\end{claim}

\begin{proof}
	If by contradiction $\mathcal{A}_{q_{1},q_{2}}^{s}$ includes $Z\otimes I$
	so $B_{p}$ and $Z\otimes I$ commute (lemma \ref{lemma:alg-commute}).
	Therefore, the algebra $\mathcal{A}_{q_{1}}^{p}$ commutes with $Z$
	(again by lemma \ref{lemma:alg-commute}) and so $[Z,P]=0$ which
	cannot be according to Claim \ref{claim:commute-entails-classical} applied to $q_1$. 

	Similarly, if by contradiction $I\otimes Z\in\a_{q_{1},q_{2}}^{s}$
	then $\a_{q_{2}}^{p}$ commute with $Z$. This implies that $\a_{q_{2}}^{p}\subseteq\left\langle Z\right\rangle $
	and so $q_{2}$ is classical because all induced algebras on it are contained in $\langle Z \rangle$ - and hence it must be trivial following 
Claim \ref{claim:no-classical-qubits} - contradiction.
\end{proof}

Observe that since $I\otimes I\in\mathcal{A}_{q_{1},q_{2}}^{s}$ so
$h\in\mathcal{A}_{q_{1},q_{2}}^{s}$ whenever $h+\lambda I\otimes I\in\mathcal{A}_{q_{1},q_{2}}^{s}$
($\lambda\in\mathbb{C})$. Also, $h\in\mathcal{A}_{q_{1},q_{2}}^{s}$
whenever $\lambda h\in\mathcal{A}_{q_{1},q_{2}}^{s}$. Therefore in
all of the following calculations, we will ignore the identity element (freely subtract from it $\lambda I\otimes I$)
and rescale the operator (multiply it by a scalar) without mentioning
it. 
\begin{claim}
	\label{claim:strongcommutation-w-not-0} $w\ne0$
\end{claim}

\begin{proof}
	Note that:
	\begin{equation}
	T^{2}=vwZ\otimes I+uwI\otimes Z+uvZ\otimes Z
	\end{equation}
	is also in $\mathcal{A}_{q_{1},q_{2}}^{s}$. By looking at $T^{2}$
	we see that if $w=0$ we get $Z\otimes Z \in\mathcal{A}_{q_{1},q_{2}}^{s}$ (because $u,v\ne 0$ by Claim \ref{claim:strongcommutation-no-trivial})
which is a contradiction to our assumption. 
\end{proof}
\begin{claim}
	\label{claim:strongcommutation-uvw-are-not-0} $u,v,w$ are all non-zero.
\end{claim}

\begin{proof}
	By claim \ref{claim:strongcommutation-w-not-0} $w\ne0$. In addition,
	$u,v$ cannot be both zero since then $Z\otimes Z\in\a_{q_{1},q_{2}}^{s}$
	(as can readily be seen by the definition of $T$). Thus suppose one of
	$u,v$ is zero and the other isn't, but then since $T^2$ 
is in $\mathcal{A}_{q_{1},q_{2}}^{s}$, 
we derive a contradiction to Claim \ref{claim:strongcommutation-no-trivial}. 
\end{proof}
\begin{claim}
	\label{claim:strongcommutation-uvw-plusminus1}
 $u^{2}=v^{2}=w^{2}$
\end{claim}

\begin{proof}
	Write:
	
	\begin{equation}
	Q=uvT-wT^{2}=v\left(u^{2}-w^{2}\right)Z\otimes I+u\left(v^{2}-w^{2}\right)I\otimes Z
	\end{equation}
	\begin{equation}
	Q^{2}=uv\left(u^{2}-w^{2}\right)\left(v^{2}-w^{2}\right)Z\otimes Z
	\end{equation}
	By looking at $Q^{2}$ it can be seen that if $w^{2}\ne u^{2}$ and
	$w^{2}\ne v^{2}$ then $Z\otimes Z\in\a_{q_{1},q_{2}}^{s}$, using
	claim \ref{claim:strongcommutation-uvw-are-not-0}. It cannot be that
	$w^{2}=u^{2}$ but $w^{2}\ne v^{2}$ (or the other way around) since
	then $Q$ cannot be in $\mathcal{A}_{q_{1},q_{2}}^{s}$ by claim \ref{claim:strongcommutation-no-trivial}.
	It follows that $u^{2}=v^{2}=w^{2}$.
\end{proof}
We are thus left with the case where $u^{2}=v^{2}=w^{2}$ and by rescaling
$T$ we may assume each of $u,v,w$ is either 1 or -1. We show that
this leads to a contradiction which in turn implies that indeed $Z\otimes Z\in\a_{q_{1},q_{2}}^{s}$.
For simplicity, assume that $u=v=w=1$: the proof where some of them
are $-1$ is exactly the same. 
Let $R\in\a_{q_{1},q_{2}}^{p}$. According
to lemma \ref{lemma:subalgebra} $\a_{q_{1},q_{2}}^{p}\subseteq\a_{q_{1}}^{p}\otimes\a_{q_{2}}^{p}\subseteq\left\langle P\right\rangle \otimes\l(\h_{q_{2}})$
and so $R$ can be written as: 
\begin{equation}
R=I\otimes R_{1}+P\otimes R_{2}
\end{equation}
where $R_{1},R_{2}\in\mathcal{L}(\h_{q_{2}})$ are arbitrary. By lemma
\ref{lem:anti-commute} we may choose a basis for $\h_{q_{1}}$
and for $\h_{q_{2}}$ for which $P=aX+bZ$ (with $a\ne0$ because
$\left[Z,P\right]\ne0$) and $R_{1}=\lambda I+rX+tZ$ (while keeping
$Z$ in tact). Since $\a_{q_{1},q_{2}}^{s}$ and $\a_{q_{1},q_{2}}^{p}$
commute so:
\[
\left[T,R\right]=\left[Z\otimes I+I\otimes Z+Z\otimes Z,rI\otimes X+tI\otimes Z+aX\otimes R_{2}+bZ\otimes R_{2}\right]=
\]
\begin{equation}
2iaY\otimes R_{2}+2irI\otimes Y+aX\otimes\left[Z,R_{2}\right]+bZ\otimes[Z,R_{2}]+2irZ\otimes Y+2iaY\otimes\left\{ Z,R_{2}\right\} +bI\otimes\left[Z,R_{2}\right]=0\label{eq:=00005BT,R=00005D}
\end{equation}
Due to linear independence in eq. \ref{eq:=00005BT,R=00005D} we may
conclude that $aX\otimes[Z,R_{2}]=0$ and so $[Z,R_{2}]=0$ because
$a\ne0$. Therefore we have:
\begin{equation}
2iaY\otimes R_{2}+2irI\otimes Y+2irZ\otimes Y+2iaY\otimes\left\{ Z,R_{2}\right\} =0\label{eq:=00005BT,R=00005D2}
\end{equation}
So $r=0$ which implies $[Z,R_{1}]=[Z,\lambda I+rX+tZ]=0$. Since
$R$ was arbitrary we conclude that $\a_{q_{1},q_{2}}^{p}$ consist
only of elements $R=I\otimes R_{1}+P\otimes R_{2}$ where $R_{1},R_{2}\in\left\langle Z\right\rangle $.
This implies that $\a_{q_{2}}^{p}\subseteq\left\langle Z\right\rangle $
and so $q_{2}$ is classical which is a contradiction. 

We have thus proved that $\left\langle Z\otimes Z\right\rangle \subseteq
\mathcal{A}_{q_{1},q_{2}}^{s}$. From Equation \ref{eq:zz} we can deduce 
$\mathcal{A}_{q_{1},q_{2}}^{s}=\langle Z \otimes Z \rangle$. 
We now further assume that $\dim\a_{q_{2}}^{p}=2$  and complete 
the proof of the lemma:  

\selectlanguage{american}%
\begin{claim}
	$\a_{q_{1}}^{p}=\left\langle X\right\rangle $
\end{claim}

\begin{proof}
	The fact that $Z\otimes Z\in\a_{q_{1},q_{2}}^{s}$ implies that
		for any $R=I\otimes R_{1}+P\otimes R_{2}\in\a_{q_{1},q_{2}}^{p}$:
	\begin{equation}
	\left[Z\otimes Z,R\right]=\left[Z\otimes Z,rI\otimes X+tI\otimes Z+aX\otimes R_{2}+bZ\otimes R_{2}\right]=2irZ\otimes Y+2iaY\otimes\left\{ Z,R_{2}\right\} +bI\otimes\left[Z,R_{2}\right]=0\label{eq:=00005BZZ,R=00005D}
	\end{equation}
	Due to linear independence in eq \ref{eq:=00005BZZ,R=00005D}, $r=0$
	which implies $[Z,R_{1}]=0$. As in the proof of Claim 
  \ref{claim:strongcommutation-uvw-plusminus1}, there must be some
	choice of $R$ for which $R_{1},R_{2}$ do not both commute with $Z$.
	Choosing such an $R$ implies that $[Z,R_{2}]\ne0$ and therefore
	$b=0$ due to the linear independence of eq \ref{eq:=00005BZZ,R=00005D}.
	Thus $P=aX$ with $a\ne0$ and hence $\a_{q_{1}}^{p}=\left\langle P\right\rangle =\left\langle X\right\rangle $
\end{proof}
\begin{claim}
	If $\a_{q_{2}}^{p}$ is 2-dimensional then $\a_{q_{1},q_{2}}^{p}=\left\langle X\otimes X\right\rangle $
\end{claim}

\begin{proof}
	In $\a_{q_{2}}^{p}$ is 2-dimensional then the roles of $q_{1}$ and
	$q_{2}$ are symmetrical, and so by changing their roles we conclude
	that $\a_{q_{2}}^{p}=\left\langle X\right\rangle $ just as we proved
	that $\a_{q_{1}}^{p}=\left\langle X\right\rangle $. 
The fact that
	$\a_{q_{1},q_{2}}^{p}=\left\langle X\otimes X\right\rangle $ is simply
	because the only operators in $Sp\{I\otimes I,X\otimes I,I\otimes X,X\otimes X\}$
	which commute with $Z\otimes Z$ are those in $Sp\{I\otimes I,X\otimes X\}=\left\langle X\otimes X\right\rangle $
	(indeed $\left[Z\otimes Z,\lambda'I\otimes I+u'X\otimes I+v'I\otimes X+w'X\otimes X\right]=2iu'Y\otimes Z+2iv'Z\otimes Y=0$
	implies that $u'=v'=0$).
\end{proof}
This concludes the proof of lemma \ref{lem:strong-commutation}. 
\end{proof}

\subsection{Proof of Theorem \ref{thm:equiv}}

\begin{proof} (of Theorem \ref{thm:equiv}). 
The proof is by induction on $r$ which appears in definition \ref{def:equiv}.
Lemma \ref{lem:strong-commutation} shows the correctness for $r=2$. 
Let $s$ be a star and $(q_{1},...,q_{r})$ a copath
of qubits of $s$ which are not in the coboundary, and such that there aren't two qubits in a row in this copath which are in the boundary. 
Thus we can apply lemma \ref{lem:strong-commutation} to each pair $q_i,q_{i+1}$ 
in this copath
to conclude that $\a^s_{q_i,q_{i+1}}=\langle Z\otimes Z\rangle$. Observe
that by lemma \ref{lemma:subalgebra}, $\mathcal{A}_{q_{1},...,q_{r+1}}^{s}\subseteq\mathcal{A}_{q_{1},...,q_{r}}^{s}\otimes\mathcal{A}_{q_{r+1}}^{s}$
and also $\mathcal{A}_{q_{1},...,q_{r+1}}^{s}\subseteq\mathcal{A}_{q_{1}}^{s}\otimes\mathcal{A}_{q_{2},...,q_{r+1}}^{s}$,
where induction hypothesis tells us that $\mathcal{A}_{q_{1},...,q_{r}}^{s}=\mathcal{A}_{q_{2},...,q_{r+1}}^{s}=\left\langle Z^{\otimes r}\right\rangle .$
Therefore, any $a\in\mathcal{A}_{q_{1},...,q_{r+1}}^{s}$ can on the
one hand be written as:
\begin{equation}
\label{eq:Zinduction1}
a=u_{1}I^{\otimes(r+1)}+u_{2}I\otimes\left(Z^{\otimes r}\right)+u_{3}Z\otimes\left(I^{\otimes r}\right)+u_{4}Z^{\otimes r+1}
\end{equation}
 and on the other hand by:
\begin{equation}
\label{eq:Zinduction2}
a=w_{1}I^{\otimes(r+1)}+w_{2}\left(I^{\otimes r}\right)\otimes Z+w_{3}\left(Z^{\otimes r}\right)\otimes I+w_{4}Z^{\otimes r+1}
\end{equation}
By joining equations \ref{eq:Zinduction1} and \ref{eq:Zinduction2}, we conclude with a linear independence argument (using $r\ge 2$) 
that $u_{2}=u_{3}=w_{2}=w_{3}=0$
which means that $a\in\left\langle Z^{\otimes r+1}\right\rangle $.
Thus $\mathcal{A}_{q_{1},...,q_{r+1}}^{s}\subseteq\left\langle Z^{\otimes r+1}\right\rangle =Sp\{I^{\otimes r+1},Z^{\otimes r+1}\}$
but since $\mathbb{C}\cdot I\subsetneq\mathcal{A}_{q_{1},...,q_{r+1}}^{s}$
so indeed $\mathcal{A}_{q_{1},...,q_{r+1}}^{s}=\left\langle Z^{\otimes r+1}\right\rangle $.
By lemma \ref{lem:strong-commutation} the algebra induced by any plaquette on each pair of adjacent qubits which are both not in the boundary is $\langle X\otimes X \rangle$ provided that at least one of them is not in the coboundary. Therefore the induction above holds for plaquettes just as well. 
\end{proof}

 \section{Proof of main lemma: Lemma \ref{lem:either-path-or-copath}}
\label{apx:proof_of_lem:either-path-or-copath}

\begin{defn} [ribbon]
\label{def:ribbon}
A ribbon is a sequence of edges $\left(e_{0},...,e_{m}\right)$ such
that each pair $e_{i-1},e_{i}$ (with $i=1,...,m$) is shared by a star $s_{i}$ and a
plaquette $p_{i}$. (see Figure \ref{fig:ribbon})
\end{defn}

The idea behind using ribbons is that a ribbon contains both a path
and a copath. This allows us to maintain (at the moment) some ambiguity
regarding the question of whether we shall later choose a star and connect to the
coboundary via a path, or choose an adjacent plaquette and connect it to the boundary via a copath. Indeed if we let $s_i$,$p_i$ denote the stars and plaquettes of the ribbon as in definition \ref{def:ribbon} and let $s_0$,$p_0$ be the other star and plaquette, besides $s_1,p_1$ which act on $e_0$ then we may consider  the path $\gamma=(s_{0},...,s_{m})$ and the copath $\gamma^{*}=\left(p_{0},...,p_{m}\right)$. To be precise, these sequences may be redundant
in the sense that e.g $\gamma$ may include several stars in a row which are
in fact identical, yet it is clear how by ignoring such multiplicities
we can think of $\gamma$ and $\gamma^{*}$ as a path and copath in a well defined sense. Note that the edges of $\gamma$ as well as the edges
of $\gamma^{*}$ belong to the ribbon. (see figure \ref{fig:ribbon})

Conversely, any path can be completed to a ribbon. Indeed let $\gamma=\left(v_{0},...,v_{m}\right)$ be a path
with corresponding edges $(\tilde{e}_{1},...,\tilde{e}_{m})$ . For
each $1\leq i\leq m-1$, connect $\tilde{e}_{i}$ and $\tilde{e}_{i+1}$
via a copath ($e_{i}^{0},...,e_{i}^{r_{i}})$ where $e_{i}^{0}=\tilde{e}_{i}$,
$e_{i}^{r_{i}}=\tilde{e}_{i+1}$ and where each $e_{i}^{j}$ has $v_{i}$
as one of its endpoint. This can be done due to the fact that $\s$ is topologically a surface and hence every point, in particular any vertex of $\k$, has a neighborhood which is either homeomorphic to the plane $\mathbb{R}^2$ or else to the upper plane (see Def \ref{def:complex}). In any event, in the plaquettes which intersect this neighborhood we can form a copath between the two edges. Thus by taking 
$\left(\tilde{e}_{1},...,\tilde{e}_{m}\right)$
and adding between $\tilde{e}_{i},\tilde{e}_{i+1}$ the 
edges $e_{i}^{1},...,e_{i}^{r_{i}-1}$
we obtain a ribbon. 

So far we haven't assumed that $\k$ (or $\s$) are connected. 
However clearly
this can be assumed since otherwise we consider each connected component
separately because there is no interaction at all between them. Note that assuming
$\k$ is connected implies not only that any two edges can be connected
by a path but also by a ribbon since any path can be completed to
a ribbon. 

We are now ready to start the construction and proof for Lemma 
$\ref{lem:either-path-or-copath}$. Figure \ref{fig:ribbon} is the main figure of the construction and is helpful in visualizing the proof. 
Let $s_{0},p_{0}$ adjacent star and plaquette in the interior, with
some shared edge $e_{0}$. As we can assume $\k$ is
connected, there exists a
path starting from $e_{0}$ and ending at a boundary/coboundary edge.
Choose such a path which is not self intersecting (each vertex is
connected to a most 2 edges of the path) and complete it to a ribbon.
It could possibly be that this ribbon includes multiple edges which
are in the boundary/coboundary. So find the first edge in the ribbon
which is in the boundary/coboundary and remove all edges following
that edge. Thus we obtain a ribbon $\beta=(e_{0},...,e_{m})$, corresponding
to the qubits $(q_{0},...,q_{m})$ with all qubits being in the interior
except $q_{m}$ which is either in the boundary or in the coboundary (or both).

Let $s_{i}$ and $p_{i}$ be the star and plaquette which share the
qubits $q_{i-1}$ and $q_{i}$ ($i=1,..,m$) (again, 
it could be that $s_{i}=s_{i+1}$ or $p_{i}=p_{i+1}$). Let $s_{m+1}$
be the other star (besides $s_{m}$ which acts on $q_{i}$ and let
$p_{m+1}$ be the other plaquette (besides $p_{m}$) which acts on
$q_{m}$. Actually, $p_{m+1}$ may not exist because $q_{m}$ may
be in the boundary. If that happens, we add a 2-simplex 
to the polygonal complex with one of its faces being $q_m$ and assign 
the identity operator to it. Clearly adding such a term to the 
instance doesn't have any effect on the groundstate. 
This makes $(p_0,...,p_{m+1})$ being a well defined copath 
with corresponding qubits $(q_0,...,q_m)$.

Define $\gamma=(s_0,...,s_{m+1})$, $\gamma^{*}=(p_0,...,p_{m+1})$, and $L=\bigotimes_{q\in\gamma}X_{q}$, $L^{*}=\bigotimes_{q\in\gamma^{*}}Z_{q}$.
The choice between the two will depend on the induced algebras on
$q_{m}$. It is almost true to say that this choice depends on whether
$q_{m}$ is in the coboundary (in which case we choose $L$) or in
the boundary (in which case we choose $L^{*}$). It is true except
for the case where $q_{m}$ is both in the coboundary and in the boundary.
We will show that  this case is not problematic due to the restrictions
that it entails on the induced algebras on $q_{m}$.

\begin{figure}
\includegraphics[scale=0.3]{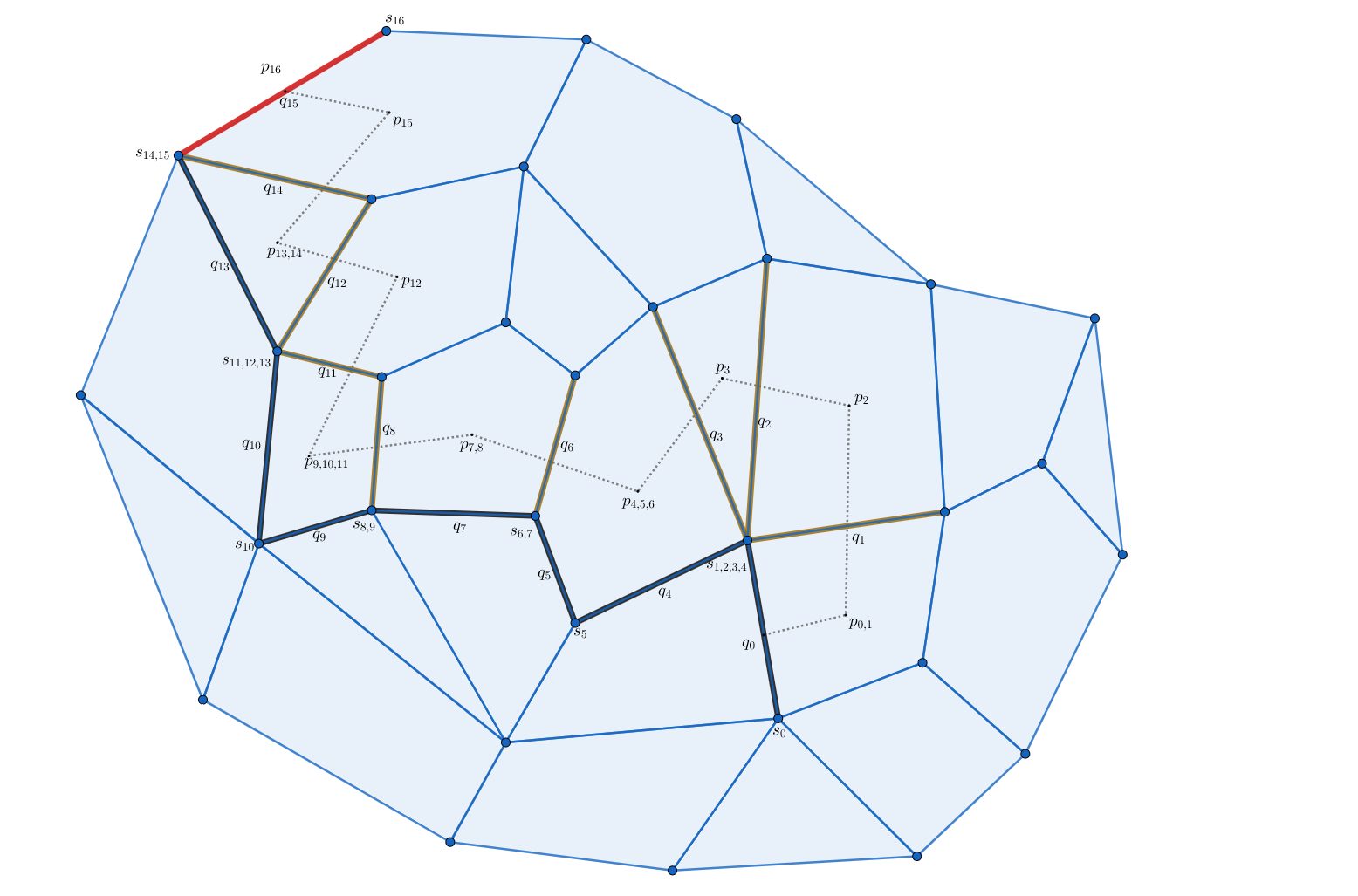}
\caption{\label{fig:ribbon} A ribbon to the boundary/coboundary: the labeled edges in bold form a ribbon. The stars and plaquettes which are shared between two adjacent edges in the ribbon are labeled as well. Ribbons include in them both a path which can easily be seen and a copath which is drawn as a dotted line. The indexes for stars/plaquettes indicate upon the possibly multiple pairs of edges in the ribbon which are acted by that star/plaquette. $q_{15}$ (or generally $q_m$) is in the boundary of the system and so the copath in this ribbon starts at $q_0$ and ends at the boundary.}
\end{figure}
 
\begin{claim}
\label{claim:properboundary1}$A_{s_{m}}$ and $B_{p_{m}}$ do not
act trivially on $q_{m}$. Therefore if $q_{m}$ is in the coboundary
then $A_{s_{m+1}}$ acts trivially on $q_{m}$ and if $q_{m}$ is
in the boundary then $B_{p_{m+1}}$ act trivially on $q_{m}$.

\end{claim}

\begin{proof}
The proof follows from the fact that $q_m$ was the first qubit in the ribbon 
to be in the boundary/coboundary. More precisely, 
note that the qubits which $A_{s_{m}}$ and $B_{p_{m}}$ share are $q_{m-1}$
and $q_{m}$. Thus if by contradiction $A_{s_{m}}$ or $B_{p_{m}}$
act trivially on $q_{m}$ then the fact that $A_{s_{m}}$ and $B_{p_{m}}$
commute implies that they commute on $q_{m-1}$ alone i.e that the
algebras $\a_{q_{m-1}}^{s_{m}}$ and $\a_{q_{m-1}}^{p_{m}}$ commute.
But $q_{m-1}$ is in the interior and so the latter implies that $q_{m-1}$
is classical by Claim \ref{claim:commute-entails-classical} which is a contradiction since we have replaced all classical qubits with trivial ones, 
but $q_{m-1}$ was assumed to be in the interior (and thus not trivial). 
Consequently, if $q_{m}$ is
in the coboundary then it must be $A_{s_{m+1}}$ which acts trivially
on $q_{m}$ and if it is in the boundary then it must be $B_{p_{m+1}}$
which acts trivially on $q_{m}$.
\end{proof}
\begin{lem}
\label{lem:properboundary2}
$\a_{q_{m}}^{s_{m}}$, $\a_{q_{m}}^{p_{m}}$ cannot be both full operator
algebras.
\end{lem}

\begin{proof}
Assume without loss of generality that $A_{q_{m}}^{s_{m}}$ is full. 
We show that this implies that $\a_{q_{m}}^{p_{m}}$ is 2-dimensional.
By how $q_{m}$ is defined, the qubit $q_{m-1}$, which by definition
is also shared by $s_{m},p_{m}$, is in the interior and hence $\a_{q_{m-1}}^{s_{m}}=\left\langle Z\right\rangle ,\a_{q_{m-1}}^{p_{m}}=\left\langle X\right\rangle $ (by applying lemma $\ref{lem:strong-commutation}$ to $q_{m-2},q_{m-1}$).
Therefore, any non zero operator in $\a_{q_{m-1},q_{m}}^{s_{m}}\subseteq\a_{q_{m-1}}^{s_{m}}\otimes\a_{q_{m}}^{s_{m}}=\left\langle Z\right\rangle \otimes\mathcal{L}(\mathcal{H}_{q_{m}})$
is of the form:
\[
T=I\otimes T_{1}+Z\otimes T_{2}
\]
 for some $T_{1},T_{2}\in\l(\h_{q_{m}})$ not both zero. We may choose
$T$ such that $T_{1}\ne0$ since otherwise we may replace $T$ with
$T^{\prime}=T^{2}+T=I\otimes\left(T_{2}\right)^{2}+Z\otimes T_{2}$
which is also in $\a_{q_{m-1},q_{m}}^{s_{m}}$ 

We may actually choose $T$ such that $T_{1}\notin\mathbb{C}\cdot I$:
assume by contradiction that $T_{1}\in\mathbb{C}\cdot I$ for any
choice of $T$. So let $T=\lambda_{1}I\otimes I+Z\otimes T_2, T'=\lambda_{2}I\otimes I+Z\otimes T^\prime_2\in\a_{q_{m-1},q_m}^{s_{m}}$
where $T,T'$
do not commute. There must exist such $T_{2},T^\prime_2$ for otherwise
we would obtain that $\a_{q_{m}}^{s_{m}}$  at most two dimensional 
(choose a basis for which $T_2,T^\prime_2\in \langle Z \rangle$ 
and conclude that $\a_{q_{m}}^{s_{m}}\subseteq \langle Z \rangle$).
Now since $I\otimes I$ is in the algebra so we obtain that also $Z\otimes T_{2}$ and $\,Z\otimes T^\prime_2$
are in the algebra and hence their product too: $I\otimes T_{2}T^\prime_{2}$.
It cannot be that $T_{2}T^\prime_{2}\in\mathbb{C}\cdot I$ because $T,T^\prime$
do not commute (indeed, if $T_{2}\cdot T^\prime_{2}=0$ so $T_{2},T^\prime_{2}$
commute, and if $T_{2}\cdot T^\prime_{2}=\lambda I$ for some $\lambda\ne0$
so $T_{2}^{-1}=\lambda^{-1}T^\prime_{2}$ and clearly any invertible operator
commutes with its inverse). 

Furthermore, we may even choose $T$ such that $T_1$ (and therefore also $T$) is
traceless: first choose $T=I\otimes T_1+Z\otimes T_2$ such that $T_1\notin\mathbb{C}\cdot I$.
Then replace $T_1$ by $T_1-tr(T_1)$. This operator is of course traceless
and non-zero. We may choose a basis
for $\h_{q_{m}}$ for which $T_1=\lambda Z$ with $\lambda\ne0$. Thus
we may choose $T$ such that $T_1=Z$ because we can divide by $\lambda$,
so we have in hand $T=I\otimes Z+Z\otimes T_{2}\in\a_{q_{m-1},q_{m}}^{s_{m}}$.

Now, by lemma \ref{lem:strong-commutation} $\a_{q_{m-1}}=\langle X\rangle$, and so every element of $\a_{q_{m-1},q_{m}}^{p_{m}}$ is of the form:
\[
R=I\otimes R_1+X\otimes R_2
\]
So
\[
\left[T,R\right]=I\otimes\left[Z,R_1\right]+X\otimes[Z,R_2]+Z\otimes[T_2,R_1]+Y\otimes\{T_2,R_2\}=0
\]
due to linear independence we conclude that $\left[Z,R_1\right]=\left[Z,R_2\right]=0$.
Since $R$ is arbitrary we conclude that $\a_{q_{m-1},q_{m}}^{p_{m}}$
consists only of elements of the form $I\otimes R_1+X\otimes R_2$ with
$R_1,R_2\in\left\langle Z\right\rangle $. This implies that $\a_{q_{m}}^{p_{m}}\subseteq\left\langle Z\right\rangle $.
By claim \ref{claim:properboundary1} $\a_{q_{m}}^{p_{m}}$ is not
trivial and so the equality holds. 
\end{proof}

\begin{cor}
\label{cor:properboundary2}
At least one of the algebras $\a_{q_{m}}^{s_{m}},\a_{q_{m}}^{p_{m}}$ is 2-dimensional.
If $\a_{q_{m}}^{p_{m}}$ is 2-dimensional, then 
 $\a_{q_{m-1,}q_{m}}^{p_{m}}=\left\langle X\otimes X\right\rangle $
If $\a_{q_{m}}^{s_{m}}$ is 2-dimensional, then $\a_{q_{m-1},q_{m}}^{s_{m}}=\left\langle Z\otimes Z\right\rangle $. 
\end{cor}

\begin{proof}
By lemma \ref{lem:properboundary2} $\a_{q_{m}}^{s_{m}},\a_{q_{m}}^{p_{m}}$ 
cannot both be full
operator algebras. By Claim \ref{claim:properboundary1} they are
both non-trivial.
It follows that one of them is 2-dimensional while
the other is at least 2-dimensional. Thus by lemma $\ref{lem:strong-commutation}$ we have that $\a_{q_{m-1},q_{m}}^{s_{m}}=\left\langle Z\otimes Z\right\rangle $
if $\a_{q_{m}}^{s_{m}}$ is 2-dimensional, 
and $\a_{q_{m-1,}q_{m}}^{p_{m}}=\left\langle X\otimes X\right\rangle $
if $\a_{q_{m}}^{p_{m}}$ is 2-dimensional. 

\end{proof}

Note that by Corollary \ref{cor:properboundary2}, if $q_m$ is in the coboundary but not in the boundary then it must be that $\a_{q_{m-1,}q_{m}}^{p_{m}}=\left\langle X\otimes X\right\rangle $ and if $q_m$ is in the boundary but not in the coboundary then $\a_{q_{m-1},q_{m}}^{s_{m}}=\left\langle Z\otimes Z\right\rangle $. But even if $q_m$ is both in the boundary and in the coboundary, still we have a demand on at least one of the algebras $\a_{q_{m-1,}q_{m}}^{p_{m}},\a_{q_{m-1},q_{m}}^{s_{m}}$. We now show that either $L=\bigotimes_{q\in\gamma}X_{q}$ 
or $L^{*}=\bigotimes_{q\in\gamma^{*}}Z_{q}$
satisfy the commutation relations which are stated in lemma \ref{lem:either-path-or-copath}
(and thus serve as the desired logical operator) depending on the which of the above algebras is 2 dimensional (if both hold, then any choice will do).

\subsubsection*{\textbf{Case a}: $\a^{p_m}_{q_{m-1},q_m}=\langle X\otimes X \rangle$} In this case we choose $s_0$ and $L=\bigotimes_{q\in\gamma}X_{q}$.

\begin{claim}
\label{claim:plaquettes-commute-with-L}$\left[B_{p},L\right]=0$
for any $p$.
\end{claim}

\begin{proof}
Aside from $q_{m}$, every qubit in $\gamma$ is not a boundary/coboundary qubit.
Therefore, the induced algebras on any $q_{i}$ ($0\leq i\leq m-1$)
by any one of the plaquettes 
it belongs to is equal to $\left\langle X\right\rangle $
and so this plaquette term of course commutes with $L$. 
Also $\a_{q_{m}}^{p_{m+1}}=\left\langle X\right\rangle $
if $q_{m}$ is not in the boundary because $\a_{q_{m}}^{p_{m}}=\left\langle X\right\rangle$ (by e.g Claim \ref{claim:interior-2-dim} and Lemma \ref{lemma:qubit-algebra}).
But also if $q_{m}$ is in the boundary then 
$B_{p_{m+1}}$ acts trivially on $q_{m}$ according to 
Lemma \ref{claim:properboundary1},
so in any case $\left[B_{p_{m+1}},L\right]=0$. Since we covered all
plaquettes which share an edge with $\gamma$ it follows that 
$\left[B_{p},L\right]=0$ for any $p.$
\end{proof}
\begin{claim}
\label{claim:L-anticommutes}$\left\{ A_{s_{0}},L\right\} =0$
\end{claim}

\begin{proof}
Since $s_{0}$ is in the interior, then $A_{s_{0}}\in\alpha Z^{\otimes|s_{0}|}$
for some $\alpha\ne0$ (we again ignore the possible addition of identity operator since it doesn't change the eigenstates of $H$). Since $\gamma$ is not self intersecting so
$q_{0}$ is the only qubit of $s_{0}$ which $\gamma$ passes through.
Consequently $\left\{ A_{s_{0}},L\right\} =\left\{ Z,X\right\} =0$.
\end{proof}
\begin{claim}
\label{claim:stars-commute-with-L}$\left[A_{s},L\right]=0$ for any
star $s\ne s_{0}.$
\end{claim}

\begin{proof}
The fact that $\left[A_{s_{m+1}},L\right]=0$ holds simply because
$A_{s_{m+1}}$ acts trivially on $q_{m}$ by claim \ref{claim:properboundary1}
and thus on any other qubit of $\gamma$. Now, given two adjacent
edges of $\gamma$, namely $q_{i},q_{j}$ ($0\leq i<j\leq m-1)$,
consider the copath $\left(q_{i},q_{i+1},...,q_{j-1},q_{j}\right).$
Therefore $s_{i}=s_{i+1}=...=s_{j-1}$ is the star acting on both
$q_{i}$ and $q_{j}$ (it could be that $i=j-1$). Since this
copath consists of qubits which are all in the interior, so by
theorem \ref{thm:equiv} $\a_{q_{i},q_{i+1},...,q_{j-1},q_{j}}^{s_{i}}=\left\langle Z^{\otimes\left(j-i+1\right)}\right\rangle $
and in particular $\a_{q_{i},q_{j}}^{s_{i}}=\langle Z\otimes Z\rangle$. Since $\gamma$ is
simple, it intersects with $s_{i}$ on those two qubits only and so
$[A_{s_{i}},L]=[Z\otimes Z,X\otimes X]=0$. 

It is thus left to prove that $\left[A_{s_{m}},L\right]=0$ since
then we have covered all stars which share an edge with $\gamma$.
This commutation is not as trivial as for the other ones since $A_{s_{m}}$
may act in various ways. However, in order to show that $A_{s_{m}}$
commutes with $L$ all we need to show is that $\a_{q_{m'},q_{m}}^{s_{m}}$
commutes with $\a_{q_{m'},q_{m}}^{L}=\left\langle X\otimes X\right\rangle $, 
or simply that it commutes with $X\otimes X,$ where $q_{m'}$ is
the other edges of $\gamma$ which $s_{m}$ acts on. To show this
we make use of the fact that $A_{s_{m}}$ must commute with all plaquette
terms. Indeed consider the copath $(q_{m'},q_{m'+1},....,q_{m-1},q_{m})$. Each
of those qubits, except $q_{m}$, is in the interior. Thus by lemma \ref{lem:strong-commutation} , $\a_{q_{j},q_{j+1}}^{p_{j}}=\left\langle X\otimes X\right\rangle $
($m'\leq j\leq m-1)$ and so the fact that $A_{s_{m}}$ commutes with
$B_{p_{j}}$ implies that it commutes with $X_{q_{j}}\otimes X_{q_{j+1}}$
(we here added the subindex just clarify the qubit on which the operator
acts on). Since this is true for any $m'\leq j\leq m-1$ then it is
true also for the product $\prod_{j=m'}^{m-1}X_{q_{j}}\otimes X_{q_{j+1}}=X_{q_{m'}}\otimes X_{q_{m}}$.
Since $q_{m'}$ and $q_{m}$ are the only qubits $B_{s_{m}}$ and
$L$ share so it follows that $\left[B_{s_{m}},L\right]=0$.
\end{proof}

\subsubsection*{\textbf{Case b}: $\a^{p_m}_{q_{m-1},q_m}\ne \langle X\otimes X \rangle$} In this case we know by Corollary \ref{cor:properboundary2} that $\a^{s_m}_{q_{m-1},q_m}= \langle Z\otimes Z \rangle$  and we choose $p_0$ and $L^{*}$.

By taking the dual (interchanging boundary with coboundary,
star with plaquette, $\gamma$ with $\gamma^{*}$, and $L$ with $L^{*}$)
we obtain claims analogous to claims \ref{claim:plaquettes-commute-with-L},\ref{claim:L-anticommutes},\ref{claim:stars-commute-with-L}.
Yet for the sake of completeness we review the proof again for this
case.
\begin{claim}
\label{claim:commutation-relation-with-L*}$\left[A_{s},L^{*}\right]=0$
for any star $s$, $\left\{ B_{p_{0}},L^{*}\right\} =0$, and $\left[B_{p},L^{*}\right]=0$
for any $p\ne p_{0}$.
\end{claim}

\begin{proof}
Aside from $q_{m}$, each qubit of $\gamma^{*}$ is in the interior.
Therefore the induced algebras on it by any of the two stars it belongs
to is $\left\langle Z\right\rangle $ and $L^{*}$ commutes with those
terms. The only star left to consider is $s_{m+1}$. If $q_{m}$ is
not in the coboundary then also $\a_{q_{m}}^{s_{m+1}}=\left\langle Z\right\rangle $.
Otherwise $A_{s_{m+1}}$ acts trivially on $q_{m}$ according to claim
\ref{claim:properboundary1}. In any event, $\left[A_{s_{m+1}},L\right]=0$.
The fact that $\left\{ B_{p_{0}},L^{*}\right\} =0$ follows from the
fact that $B_{p_{0}}$ is in the interior and thus up to a constant
it is equal to $X^{\otimes|p_{0}|}$. The fact that $\left[B_{p_{m+1}},L^{*}\right]=0$
is simply because $B_{p_{m+1}}$ acts trivially on $q_{m}$. Now given
two adjacent edges of $\gamma^{*}$, $q_{i},q_{j}$, we have the path
$\left(q_{i},q_{i+1},...,q_{j-1},q_{j}\right)$ of qubits which are
all in the interior and so in particular $\a_{q_{i},q_{j}}^{p_{i}}=\left\langle X\otimes X\right\rangle $.
Regarding $B_{p_{m}},$ let $q_{m'}$ be the edge coming before $q_{m}$
on $\gamma^{*}$. So $\left(q_{m'},...,q_{m}\right)$ is a path of
edges of $B_{p_{m}}$ which are all in the interior except $q_{m}$.
It follows that $\a_{q_{j},q_{j+1}}^{s_{j}}=\left\langle X_{q_{j}}\otimes X_{q_{j+1}}\right\rangle $
and so $\left[B_{p_{m}},X_{q_{j}}\otimes X_{q_{j+1}}\right]=0$ for
any $m'\leq j\leq m-1$. This implies that $\left[B_{p_{m}},L^{*}\right]=\left[B_{p_{m}},X_{q_{m'}}\otimes X_{q_{m}}\right]=0$.
\end{proof}

\section{Puncturing the surface to obtain 2-locality}
\label{apx:puncturing_2local}

\subsection{Proof of Claim \ref{claim:triangle-center}}
Choose some ball $B$ of radius $2k$ contained in $T$ and 
let $p$ be a plaquette in $H$ 
which includes the vertex which is the center of $B$, 
and let $s$ be a star in $H$ 
adjacent to $p$, namely sharing with it 
$2$ edges. Due to the fact that the radius of $B$ is $2k$, $B$ 
contains all edges of $s$ and $p$, as they are all within distance 
less than $2k$ from the center of $B$.  
Now, we claim that at least one of $s$ and $p$ are not in the 
interior (namely, act trivially on one of their qubits), 
or in other words, if they are both in the interior, at least one of them  
is in $\mathcal{W}$ and therefore was removed from $\tilde{H}$; this is because 
by Lemma \ref{lem:either-path-or-copath} at least one 
of them has access to the boundary/coboundary, and thus was  
removed from $H$ in the construction of $\tilde{H}$. 
We let $h$ be that term.

\subsection{Proof of Claim \ref{claim:super-particle}}

We first show a constant upper bound on the number of 
qubits that each triangle contains. 
For this we recall that each triangle in contained in a 
ball of radius $R$ in the $1$-skeleton of $\k$.  
The following fact can be found in \cite{Graph2}.

\begin{fact}[Moore's bound]
\label{fact:moore}
Consider a graph with maximal degree at most $k$ and diameter $R$ (diam between points with no path is $\infty$). Then if $V$ is the set of 
vertices in the graph, we have $|V|\leq N_{k,R}$ 
where $N_{k,R}$ is the Moore bound:
$$ N_{k,R}=1+k\sum_{i=0}^{R-1}(k-1)^{i}$$
\end{fact}
\begin{cor}
Under the same condition on the graph as in Moore's bound, 
the number of edges in the graph satisfies
 $$ |E|\leq \frac{1}{2}k\cdot N_{k,R}. $$
\end{cor}
\begin{proof}
Let $\Sigma$ denote the sum of all degree of vertices of $G$. It follows that: 
$$ |E|\leq \frac{1}{2}\Sigma \leq \frac{1}{2} k\cdot |V|\leq \frac{1}{2}k\cdot N_{k,R}  $$
\end{proof}

\begin{proof}(of Claim \ref{claim:super-particle})
Let $F$ be a face of the dual graph $G$ 
constructed above, drawn on $\s$ as described. 
Since the degree of $\mathcal{T}$ is at most $D$,  
$F$ intersects with at most $D$ triangles. Hence each super-particle 
contains at most $\frac{1}{2}Dk\cdot N_{k,R}$, and the claim follows.
\end{proof}

\subsection{Proof of Claim \ref{claim:2-local}}
\label{apx:2-local}
Every local term which doesn't intersect with the drawn edges of $G$ 
acts only on one
super-particle. Every local term which intersects with an edge of
$G$ but not a vertex of $G$ interacts with at most two super-particles.
The only terms which act on three super-particles are those which 
intersect with three faces of $G$; these are the terms $h$
corresponding to
triangle centers. However the
paths were drawn in such a way that $h$ acts on one of those super-particles trivially so it is 2-local as well.

\end{appendices}

\end{document}